\documentclass{llncs}

\usepackage{latexsym}
\usepackage{amssymb}
\usepackage{amsmath}
\usepackage{amsfonts}
\usepackage{xspace}
\usepackage{color}
\usepackage[mathscr]{euscript}
\usepackage{url}
\usepackage{float}
\usepackage{framed}

\usepackage{epsfig}
\usepackage{ifpdf}

\usepackage[a4paper,hmargin=1in,vmargin=1in]{geometry}

\newcommand{\etal}{{\em et al.}\xspace}

\newcommand{\set}[1]{\{#1\}}
\newcommand{\Set}[2]{\{ #1 \,|\, #2\}}

\newcommand{\N}{\Bbb{N}}

\newcommand{\F}{\Bbb{F}}
\newcommand{\C}{\Bbb{C}}
\renewcommand{\H}{\mathscr{H}}

\newcommand{\dens}{{\cal D}}

\renewcommand{\phi}{\varphi}

\newcommand{\ket}[1]{{\ensuremath{\lvert#1\rangle}}}
\newcommand{\bra}[1]{{\ensuremath{\langle#1\rvert}}}

\newcommand{\ketbra}[2]{\ket{#1}\hspace{-0.3ex}\bra{#2}}
\newcommand{\proj}[1]{\ketbra{#1}{#1}}

\newcommand{\tr}{\mathrm{tr}}

\newcommand{\guess}{\mathrm{Guess}}

\newcommand{\I}{\Bbb{I}}
\newcommand{\id}{id}
\newcommand{\UNIF}{U}

\newcommand{\MAC}{{\sf MAC}}
\newcommand{\QMAC}{{\sf QMAC}}
\newcommand{\QENC}{{\sf QENC}}
\newcommand{\Auth}{{\sf Auth}}
\newcommand{\Verify}{{\sf Verify}}
\newcommand{\Refresh}{{\sf Refresh}}

\newcommand{\hash}{{\sf H}}
\newcommand{\Ext}{{\sf Ext}}
\renewcommand{\SS}{{\sf SS}}

\newcommand{\Exe}{{\cal E}\hspace{-0.3ex}{\sl x\hspace{-0.15ex}e}}

\newcommand{\K}{{\cal K}}
\newcommand{\KEY}{{\cal K\hspace{-0.15ex}E\hspace{-0.2ex}Y}}
\newcommand{\MSG}{{\cal M\hspace{-0.15ex}S\hspace{-0.1ex}G}}
\newcommand{\TAG}{{\cal T}}

\newcommand{\msg}{m\hspace{-0.2ex}s\hspace{-0.15ex}g}
\newcommand{\key}{k\hspace{-0.1ex}e\hspace{-0.15ex}y}

\newcommand{\rA}{\ensuremath{\textit{\textsf{A}}}}
\newcommand{\rB}{\ensuremath{\textit{\textsf{B}}}}
\newcommand{\rC}{\ensuremath{\textit{\textsf{C}}}}
\newcommand{\rE}{\ensuremath{\textit{\textsf{E}}}}

\newcommand{\BB}{\text{\raisebox{0.3ex}{\tiny \sc bb84}}}

\newcommand{\code}{{\cal C}}
\newcommand{\noise}{\phi}

\newcommand{\eps}{\varepsilon}
\newcommand{\epsclose}[1][\varphi]{\approx_{#1}}

\newenvironment{mybox}
          {\smallskip\noindent\hspace{-0.4ex}\begin{minipage}{1.01\columnwidth} \begin{framed}\hspace{-1ex}\begin{minipage}{1.01\columnwidth} }
         {\vspace{-2.6ex} \end{minipage} \end{framed} 
\end{minipage}\par\bigskip}

\allowdisplaybreaks

\title{Quantum Authentication and Encryption with Key Recycling\thanks{A differently formatted version of (version v2 of) this arXiv article 1610.05614 is available from the proceedings of {\em Advances in Cryptology\,-\,EUROCRYPT 2017} (Springer-Verlag), or from {\tt eprint.iacr.org/2017/102}.}
}

\subtitle{Or: How to Re-use a One-Time Pad Even if $P\!=\!N\!P$\,---\,Safely\,\&\,Feasibly}

\author{Serge Fehr\inst{1}\thanks{Member of QuSoft, the Dutch research center for quantum software ({\tt www.qusoft.org}). } 
\and Louis Salvail\inst{2}\thanks{Funded by Canada's NSERC discovery grant
and NSERC discovery accelerator. }}

\institute{Cryptology Group, Centrum Wiskunde \& Informatica (CWI), Amsterdam, The Netherlands 
\and 
Department of Computer Science and Operations Research (DIRO), 
Universit{\'e} de Montr{\'e}al, Canada 
}

\begin{document}

\maketitle


\setcounter{footnote}{0}

\begin{abstract}
We propose an information-theoretically secure encryption scheme for classical messages with quantum ciphertexts that offers {\em detection} of eavesdropping attacks, and {\em re-usability of the key} in case no eavesdropping took place: the entire key can be securely re-used for encrypting new messages as long as no attack is detected. 
This is known to be impossible for fully classical schemes, where there is no way to detect plain eavesdropping attacks. 

This particular application of quantum techniques to cryptography was originally proposed by Bennett, Brassard and Breidbart in 1982, even before proposing quantum-key-distribution, and a simple candidate scheme was suggested but no rigorous security analysis was given.  
The idea was picked up again in 2005, when Damg{\aa}rd, Pedersen and Salvail suggested a new scheme for the same task, but now with a rigorous security analysis. However, their scheme is much more demanding in terms of quantum capabilities: it requires the users to have a quantum computer. 

In contrast, and like the original scheme by Bennett \etal, our new scheme requires from the honest users merely to prepare and measure single BB84 qubits. 
As such, we not only show the first provably-secure scheme that is within reach of current technology, but we also confirm Bennett \etal's original intuition that a scheme in the spirit of their original construction is indeed secure. 
\end{abstract}

\section{Introduction}

\paragraph{\sc Background. } 

Classical information-theoretic encryption (like the one-time pad) and authentication (like Carter-Wegman authentication) have the serious downside that the key can be re-used only a small number of times, e.g. only {\em once} in case of the one-time pad for encryption or a strongly universal$_2$ hash function for authentication. This is inherent since by simply {\em observing} the communication, an eavesdropper Eve inevitably learns a substantial amount of information on the key. Furthermore, there is no way for the communicating parties, Alice and Bob, to {\em know} whether Eve is present and has observed the communication or not, so they have to assume the worst. 

This situation changes radically when we move to the quantum setting and let the ciphertext (or authentication tag) be a quantum state: then, by the fundamental properties of quantum mechanics, an Eve that {\em observes} the communicated state inevitably {\em changes} it, and so it is potentially possible for the receiver Bob to detect this, and, vice versa, to conclude that the key is still secure and thus can be safely re-used in case everything looks as it is supposed to be. 

This idea of key re-usability by means of a quantum ciphertext goes back to a manuscript titled ``{\em Quantum Cryptography II: How to re-use a one-time pad safely even if $P\!=\!N\!P$}'' by Bennett, Brassard and Breidbart written in 1982. However, their paper was originally not published, and the idea was put aside after two of the authors discovered what then became known as BB84 quantum-key-distribution~\cite{BB84}.%
\footnote{A freshly typeset version of the original manuscript was then published more than 30 years later in~\cite{BBB}. }
Only much later in 2005, this idea was picked up again by Damg{\aa}rd, Pedersen and Salvail in~\cite{DPS05} (and its full version in~\cite{DPS14}), where they proposed a new such encryption scheme and gave a rigorous security proof\,---\,in contrast, Bennett \etal's original reasoning was very informal and hand-wavy.  

The original scheme by Bennett \etal is simple and natural: you one-time-pad encrypt the message, add some redundancy by encoding the ciphertext using an error correction (or~detection) code, and encode the result bit-wise into what we nowadays call BB84 qubits. The scheme by Damg{\aa}rd \etal is more involved; in particular, the actual quantum encoding is not done by means of single qubits, but by means of states that form a set of mutually unbiased bases in a Hilbert space of large dimension. This in particular means that their scheme requires a {\em quantum computer} to produce the quantum ciphertexts and to decrypt them.

\paragraph{\sc Our Results. } 

We are interested in the question of whether one can combine the simplicity of the originally proposed encryption scheme by Bennett \etal with a rigorous security analysis as offered by Damg{\aa}rd \etal for their scheme; in particular, whether there is a provably secure scheme that is within reach of being implementable with current technology\,---\,and we answer the question in the affirmative.  

We start with the somewhat simpler problem of finding an {\em authentication} scheme that allows to re-use the key in case no attack is detected, and we show a very simple solution. In order to authenticate a (classical) message $\msg$, we encode a random bit string $x \in \set{0,1}^n$ into BB84 qubits $H^\theta\ket{x}$, where $\theta \in \set{0,1}^n$ is part of the shared secret key, and we compute a tag $t = \MAC(k,\msg\|x)$ of the message concatenated with $x$, where $\MAC$ is a classical information-theoretic one-time message authentication code, and its key $k$ is the other part of the shared secret key. The qubits $H^\theta\ket{x}$ and the classical tag $t$ are then sent along with $\msg$, and the receiver verifies correctness of the received message in the obvious way by measuring the qubits to obtain $x$ and checking $t$. 

One-time security of the scheme is obvious, and the intuition for key-recycling is as follows. Since Eve does not know $\theta$, she has a certain minimal amount of uncertainty in $x$, so that, if $\MAC$ has suitable extractor-like properties, the tag $t$ is (almost) random and {\em independent} of $k$ and $\theta$, and thus gives away no information on $k$ and $\theta$. Furthermore, if Eve tries to gain 
information on $k$ and $\theta$ by measuring some qubits, she disturbs these qubits and is likely to be detected. 
A subtle issue is that if Eve measures only {\em very few} qubits then she has a good chance of not being detected, while still learning a little bit on $\theta$ by the fact that she has not been detected. However, as long as her uncertainty in $\theta$ is large enough this should not help her (much), and the more information on $\theta$ she tries to collect this way the more likely it is that she gets caught. 

We show that the above intuition is correct. Formally, we prove that as long as the receiver Bob accepts the authenticated message, the key-pair $(k,\theta)$ can be safely re-used, and if Bob rejects, it is good enough to simply refresh~$\theta$. Our proof is based on techniques introduced in~\cite{TFKW13} and extensions thereof. 

Extending our authentication scheme to an encryption scheme is intuitively quite easy: we simply extract a one-time-pad key from $x$, using a strong extractor (with some additional properties) with a seed that is also part of the shared secret key. Similarly to above, we can prove that as long as the receiver Bob accepts, the key can be safely re-used, and if Bob rejects it is good enough to refresh~$\theta$. 

In our scheme, the description length%
\footnote{In our scheme, $\theta$ is not uniformly random in $\set{0,1}^n$ but is chosen to be a code word, as such, its description length is smaller than its physical bit length, and given by the dimension of the code. }
of $\theta$ is $m + 3\lambda$, where $m$ is the length of the encrypted message $\msg$ and $\lambda$ is the security parameter (so that the scheme fails with probability at most $2^{-\lambda}$). Thus, with respect to the number of fresh random bits that are needed for the key refreshing, i.e. for updating the key in case Bob rejects, our encryption scheme is comparable to the scheme by Damg{\aa}rd \etal%
\footnote{Their scheme needs $m+\ell$ fresh random bits for key refreshing, where $\ell$ is a parameter in their construction, and their scheme fails with probability approximately $2^{-\ell/2}$. }
and optimal in terms of the dependency on the message length~$m$. 

Our schemes can be made {\em noise robust} in order to deal with a (slightly) noisy quantum communication; the generic solution proposed in~\cite{DPS05,DPS14} of using a quantum error correction code is not an option for us as it would require a quantum computer for en- and decoding. 
Unfortunately, using straightforward error correction techniques, like sending along the syndrome of $x$ with respect to a suitable error correcting code, renders our proofs invalid beyond an easy fix, though it is unclear whether the scheme actually becomes insecure. However, we can deal with the issue by means of using error correction ``without leaking partial information'', as introduced by Dodis and Smith~\cite{DS05} and extended to the quantum setting by Fehr and Schaffner~\cite{FS08}. Doing error correction in a more standard way, which would offer more freedom in choosing the error correction code and allow for a larger amount of noise, remains an interesting open problem. 


\paragraph{\sc Encryption with Key Recycling vs QKD. } 

A possible objection against the idea of encryption with key recycling is that one might just as well use~QKD to produce a new key, rather than re-using the old one. However, there {\em are} subtle advantages of using encryption with key recycling instead.
For instance, encryption with key recycling is (almost) non-interactive and requires only {\em 1 bit} of authenticated feedback: ``accept'' or ``reject'', that can be provided {\em offline}, i.e., after the communication of the private message, as long as it is done before the scheme is re-used. This opens the possibility to provide the feedback by means of a different channel, like by confirming over the phone. 
In contrast, for QKD, a {\em large amount} of data needs to be authenticated {\em online} and in {\em both directions}.  
If no physically authenticated channel is available, then the authenticated feedback can actually be done very easily: Alice appends a random token to the message she communicates to Bob in encrypted form, and Bob confirms that no attack is detected by returning the token back to Alice\,---\,in~plain\,---\,and in case he detected an attack, he sends a reject message instead.%
\footnote{
Of course, when Bob sends back the token to confirm, Eve can easily replace it by the reject message and so prevent Alice and Bob from finding agreement, but this is something that Eve can always achieve by ``altering the last message'', also in QKD. } 
Furthermore, encryption with key recycling has the potential to be {\em more efficient} than QKD in terms of communication. Even though this is not the case for our scheme, there is certainly potential, because no sifting takes place and hence there is no need to throw out a fraction of the quantum communication. 
Altogether, on a stable quantum network for instance, encryption with key-recycling would be the preferred choice over QKD. 
Last but not least, given that the re-usability of a one-time-pad-like encryption key was one of the very first proposed applications of quantum cryptography\,---\,even before QKD\,---\,we feel that giving a satisfactory answer should be of intellectual interest.

\paragraph{\sc Related Work. } 

Besides the work of Brassard \etal and of Damg{\aa}rd \etal, who focus on encrypting {\em classical} messages, there is a line of work, like~\cite{Leung02,OH03,HLM}, that considers key recycling in the context of authentication and/or encryption of {\em quantum} messages. However, common to almost all this work is that only {\em part} of the key can be re-used if no attack is detected, or a new but {\em shorter} key can be extracted. 
The only exceptions we know of are the two recent works by Garg \etal~\cite{GYZ16} and by Portmann~\cite{Por16}, which consider and analyze authentication schemes for quantum messages that do offer re-usability of the entire key in case no attack is detected. However, these schemes are based on techniques (like unitary designs) that require the honest users to perform quantum computations also when restricting to classical messages. Actually,~\cite{Por16} states it as an explicit open problem to ``find a prepare-and-measure scheme to encrypt and authenticate a classical message in a quantum state, so that all of the key may be recycled if it is successfully authenticated''. 
On the other hand, their schemes offer security against {\em superposition 
attacks}, where the adversary may trick the sender into authenticating a {\em superposition} of classical messages; this is something we do not consider here\,---\,as a matter of fact, it would be somewhat unnatural for us since such superposition attacks require the sender (wittingly or unwittingly) to hold a quantum computer, which is exactly what we want to avoid.

\section{Preliminaries}

\subsection{Basic Concepts of Quantum Information Theory}

We assume basic familiarity; we merely fix notation and terminology here.%

\paragraph{\sc Quantum states. } 

The state of a quantum system with state space $\H$ is specified by a {\em state vector} $\ket{\phi} \in \H$ in case of a pure state, or, more generally in case of a mixed state, by a {\em density matrix} $\rho$ acting on $\H$.  The
set of density matrices acting on $\H$ is denoted $\dens(\H)$. 
We typically identify different quantum systems by means of labels $\rA,\rB$ etc., and we write $\rho_\rA$ for the state of system $\rA$ and $\H_\rA$ for its state space, etc.  
The joint state of a bipartite system $\rA\rB$ is given by a density matrix $\rho_{\rA\rB}$ in $\dens(\H_\rA \otimes \H_\rB)$; it is then understood that $\rho_{\rA}$ and $\rho_{\rB}$ are the respective {\em reduced} density matrices $\rho_{\rA} = \tr_\rB(\rho_{\rA\rB})$ and $\rho_{\rB} = \tr_\rA(\rho_{\rA\rB})$. 

\smallskip

We also consider states that consist of a classical and a quantum part. Formally, $\rho_{X\rE} \in \dens(\H_X \otimes \H_\rE)$ is called a {\em cq}-state (for {\em c}lassical-{\em q}uantum), if it is of the form
$$
\rho_{X\rE} = \sum_{x \in \cal X} P_X(x) \proj{x} \otimes \rho_\rE^x \enspace,
$$ 
where $P_X: {\cal X} \to [0,1]$ is a probability distribution, $\set{\ket{x}}_{x \in \cal X}$ is a fixed orthonormal basis of $\H_X$, and $\rho_\rE^x \in \dens(\H_\rE)$. Throughout, we will slightly abuse notation and express this by writing $\rho_{X\rE} \in \dens({\cal X} \otimes \H_\rE)$. 

In the context of such a cq-state $\rho_{X\rE}$, an {\em event} $\Lambda$ is specified by means of a decomposition $\rho_{X\rE} = P[\Lambda] \cdot \rho_{X\rE|\Lambda} + P[\neg\Lambda] \cdot \rho_{X\rE|\neg\Lambda}$ with $P[\Lambda],P[\neg\Lambda] \geq 0$ and $\rho_{X\rE|\Lambda},\rho_{X\rE|\neg\Lambda}\in \dens({\cal X} \otimes \H_\rE)$. 
Associated to such an event $\Lambda$ is the {\em indicator random variable} $1_{\!\Lambda}$, i.e., the cq-state $\rho_{X 1_{\!\Lambda} \rE} \in \dens({\cal X} \otimes \set{0,1} \otimes \H_\rE)$, defined in the obvious way. 
Note that, for any cq-state $\rho_{X\rE}$ and any $x \in \cal X$, the event $X\!=\!x$ is naturally defined and $\rho_{X\rE|X=x} = \proj{x} \otimes \rho_\rE^x$ and $\rho_{\rE|X=x} = \rho_\rE^x$. 

\smallskip

If a state $\rho_{X}$ is purely classical, meaning that $\rho_{X} = \sum_x P_X(x) \proj{x}$ and expressed as $\rho_{X} \in \dens({\cal X})$, we may refer to standard probability notation so that probabilities like $P[X\!=\!x]$ are well understood. 
Finally, we write $\mu_{\cal X}$ for the {\em fully mixed state} $\mu_{\cal X} = \frac{1}{|{\cal X}|} \sum_x \proj{x} = \frac{1}{|{\cal X}|} \I_{\cal X} \in \dens({\cal X})$.

\paragraph{\sc General quantum operations. } 

Operations on quantum systems are described by {\em CPTP maps}. To emphasize that a CPTP map $\mathcal{Q}: \dens(\H_\rA) \to \dens(\H_{\rA'})$ acts on density matrices in $\dens(\H_\rA)$, we sometimes write $\mathcal{Q}_\rA$, and we say that it ``acts on $\rA$''. Also, we may write $\mathcal{Q}_{\rA\to\rA'}$ in order to be explicit about the range too. 
If $\mathcal Q$ is a CPTP map acting on $\rA$, we often abuse notation and simply write ${\cal
Q}_\rA(\rho_{\rA\rB})$ or $\rho_{{\cal Q}(\rA)\rB}$ for $\bigl({\cal Q}_\rA \otimes \id_\rB\bigr)(\rho_{\rA\rB})$, where
$\id_\rB$ is the identity map on $\dens(\H_\rB)$. 

\smallskip

In line with our notation for cq-states, $\mathcal{Q}: \dens({\cal X} \otimes \H_\rE) \to \dens({\cal X}' \otimes \H_{\rE'})$ is used to express that $\mathcal{Q}$ maps any cq-state $\rho_{X\rE} \in \dens({\cal X} \otimes \H_\rE)$ to a cq-state $\mathcal{Q}(\rho_{X'\rE})$ in $\dens({\cal X}' \otimes \H_{\rE'})$. 
We say that a CPTP map $\mathcal{Q}: \dens({\cal X} \otimes \H_\rE) \to \dens({\cal X} \otimes \H_{\rE'})$ is ``controlled by $X$ and acts on~$\rE\,$'' if on a cq-state $\rho_{X\rE}\in \dens({\cal X} \otimes \H_\rE)$ it acts as
$$
\mathcal{Q}(\rho_{X\rE}) = \sum_x P_X(x) \proj{x}\otimes \mathcal{Q}^x(\rho^x_\rE)
$$
with ``conditional'' CPTP maps $\mathcal{Q}^x: \dens(\H_\rE) \to \dens(\H_{\rE'})$. 
Note that in this case we write $\mathcal{Q}_{X\rE \to \rE'}$ rather than $\mathcal{Q}_{X\rE \to X\rE'}$, as it is understood that $\mathcal{Q}$ keeps $X$ alive. 
For concreteness, we require that such a $\mathcal{Q}$ is of the form $\mathcal{Q} = \sum_x {\cal P}_{\proj{x}} \otimes \mathcal{Q}^x$ where ${\cal P}_{\proj{x}}(\rho) = \proj{x}\,\rho\, \proj{x}$ for any $\rho \in \dens(\H_X)$.%
\footnote{This means that the system $X$ is actually measured (in the fixed basis $\set{\ket{x}}_{x \in \cal X}$). }
As such, $\mathcal{Q}$ is fully specified by means of the conditional CPTP maps $\mathcal{Q}^x$. 
Finally, for any function $f:{\cal X} \to {\cal Y}$, we say that $\mathcal{Q}: \dens({\cal X} \otimes \H_\rE) \to \dens({\cal X} \otimes \H_{\rE'})$ is ``controlled by $f(X)$'' if it is controlled by $X$, but $\mathcal{Q}^x = \mathcal{Q}^{x'}$ for any $x,x' \in \cal X$ with $f(x) = f(x')$.

\paragraph{\sc Markov-chain states. } 

Let $\rho_{XY\rE} \in \dens({\cal X} \otimes {\cal Y} \otimes \H_\rE)$ be a cq-state with two classical subsystems $X$ and~$Y$. Following~\cite{DFSS07}, we define $\rho_{X\leftrightarrow Y\leftrightarrow\rE}$ to be the ``Markov-chain state''
$$
\rho_{X\leftrightarrow Y\leftrightarrow\rE} := \sum_{x,y} P_{XY}(x,y) \proj{x} \otimes \proj{y} \otimes \rho_\rE^y 
$$
with $\rho_\rE^y = \sum_x P_{X|Y}(x|y)\, \rho_\rE^{x,y}$. If the state $\rho_{XY\rE}$ is clear from the context we write $X\leftrightarrow Y\leftrightarrow\rE$ to express that $\rho_{XY\rE} = \rho_{X\leftrightarrow Y\leftrightarrow\rE}$. It is an easy exercise to verify that the Markov-chain condition $X\leftrightarrow Y\leftrightarrow\rE$ holds if and only if $\rho_{XY\rE} = {\cal Q}_{Y\!\varnothing\to\rE}(\rho_{XY})$ for a CPTP map ${\cal Q}_{Y\!\varnothing\to \rE}: \dens({\cal Y}) \to \dens({\cal Y} \otimes \H_\rE)$ that is controlled by $Y$ and acts on the ``empty'' system $\varnothing$, i.e., the conditional maps act as ${\cal Q}^y_{\varnothing \to \rE}: \dens(\C) \to \dens(\H_\rE)$.

\paragraph{\sc Quantum measurements. } 

We model a {\em measurement} of a quantum system $\rA$ with outcome in $\cal X$ by means of a CPTP map $\mathcal{M}: \dens(\H_\rA) \to \dens({\cal X})$ that acts as
$$
{\cal M}(\rho) = \sum_{x\in \cal X} \tr(E_x \rho) \proj{x} \enspace,
$$
where $\{\ket{x}\}_{x \in \cal X}$ is a fixed basis, and $\set{ E_x }_{x \in \cal X}$ forms a {\em POVM}, i.e., a family of positive-semidefinite operators that add up to the identity matrix $\I_{\cal X}$. A measurement ${\cal M}: \dens({\cal Z} \otimes \H_\rA) \to \dens({\cal Z} \otimes {\cal X})$ is said to be a ``measurement of $\rA$ controlled by $Z$" if it is controlled by $Z$ and acts on $\rA$ as a CPTP map. It is easy to see that in this case the conditional CPTP maps ${\cal M}^z: \dens(\H_\rA) \to \dens({\cal X})$ are measurements too, referred to as  ``conditional measurements''. 

Note that whenever ${\cal M}: \dens(\H_Z \otimes \H_\rA) \to \dens({\cal X})$ is an {\em arbitrary} measurement of $Z$ and $\rA$ that is applied to a cq-state $\rho_{Z\rA} \in \dens({\cal Z} \otimes \H_\rA)$, we may assume that ${\cal M}$ first ``produces a copy of $Z$'', and thus we may assume without loss of generality that ${\cal M}: \dens({\cal Z} \otimes \H_\rA) \to \dens({\cal Z} \otimes {\cal X})$  is {\em controlled} by $Z$. 

For a given $n \in \N$, ${\cal M}^\BB_{\Theta\rA \to X}$ denotes the {\em BB84 measurement} of an $n$-qubit system $\rA$ controlled by~$\Theta$. Formally, for every $\theta \in \set{0,1}^n$ the corresponding conditional measurement is specified by the POVM $\set{H^\theta \proj{x} H^\theta}$ with $x$ ranging over $\set{0,1}^n$. Here, $H$ is the {\em Hadamard matrix}, and $H^\theta \ket{x}$ is a short hand for $H^{\theta_1} \ket{x_1} \otimes \cdots \otimes H^{\theta_n} \ket{x_n} \in \H_\rA = (\C^2)^{\otimes n}$, where $\set{\ket{0},\ket{1}}$ is the {\em computational basis} of the qubit system $\C^2$.

\paragraph{\sc Trace Distance. } 

We capture the distance between two states $\rho,\sigma \in \dens(\H)$ in terms of their \emph{trace distance} $\delta(\rho, \sigma):= \frac12\|\rho-\sigma\|_1$, where \smash{$\|K\|_1:= \tr\bigl(\sqrt{K^\dagger K}\bigr)$} is the \emph{trace norm} of an arbitrary operator $K$. If the states $\rho_{\rA}$ and $\rho_{\rA'}$ are clear from context, we may write $\delta(\rA, \rA')$ instead of $\delta(\rho_{\rA}, \rho_{\rA'})$. Also, for any cq-state $\rho_{X\rE}$ in $\dens({\cal X} \otimes \H_\rE)$, we write $\delta(X, \UNIF_{\cal X}|\rE)$ as a short hand for $\delta(\rho_{X\rE}, \mu_{\cal X} \otimes \rho_\rE)$. Obviously, $\delta(X, \UNIF_{\cal X}|\rE)$ captures how far away $X$ is from uniformly random on $\cal X$ when given the quantum system~$\rE$. 

It is well known that the trace distance is monotone under CPTP maps, and it is easy to see that if two cq-states $\rho_{X\rE}, \sigma_{X\rE} \in \dens({\cal X} \otimes \H_\rE)$ coincide on their classical subsystems, meaning that $\rho_X = \sigma_X$, then $\delta(\rho_{X\rE}, \sigma_{X\rE})$ decomposes into 
$\delta(\rho_{X\rE}, \sigma_{X\rE}) = \sum_x P_X(x) \, \delta(\rho^x_{\rE}, \sigma^x_{\rE})$.

\subsection{The Guessing Probability}

An important concept in the technical analysis of our scheme(s) is the following notion of guessing probability, which is strongly related to the (conditional) min-entropy as introduced by Renner~\cite{Ren10}, but turns out to be more convenient to work with for our purpose. 
Let $\rho_{X\rE} \in \dens({\cal X} \otimes \H_\rE)$ be a cq-state. 

\begin{definition}\label{def:guess}
The {\em guessing probability} of $X$ given $\rE$ is 
$$
\guess(X|\rE) := \max_{\cal M} P[{\cal M}(\rE) \!=\! X] \enspace,
$$
where the maximum is over all measurements $\cal M: \dens(\H_\rE) \to \dens({\cal X})$ of $\rE$ with outcome in $\cal X$.%
\footnote{By our conventions, the probability $P[{\cal M}(\rE) \!=\! X]$ is to be understood as  $P[X' \!=\! X]$ for the (purely classical) state $\rho_{XX'} = \rho_{X{\cal M}(\rE)} = (\id_X \otimes {\cal M})(\rho_{X\rE}) \in \dens({\cal X} \otimes {\cal X})$. }
\end{definition}
Note that if $\Lambda$ is an event, then $\guess(X|\rE,\Lambda)$ is naturally defined by means of applying the above to the ``conditional state" $\rho_{X\rE|\Lambda} \in \dens({\cal X} \otimes \H_\rE)$. 

\smallskip

We will make use of the following elementary properties of the guessing probability. In all the statements, it is understood that $\rho_{X\rE} \in \dens({\cal X} \otimes \H_\rE)$, respectively $\rho_{XZ\rE} \in \dens({\cal X} \otimes {\cal Z} \otimes \H_\rE)$ in Property~\ref{prop:expansion}. 

\begin{property}\label{prop:monotonicity}
$\guess(X|{\cal Q}(\rE)) \leq \guess(X|\rE)$ for any CPTP map $\cal Q$ acting on $\rE$. 
\end{property}

\begin{property}\label{prop:expansion}
$\guess(X|Z \rE) = \sum_z P_Z(z) \, \guess(X|\rE,Z\!=\!z)$. 
\end{property}

\begin{property}\label{prop:cond}
$\guess(X|\rE,\Lambda) \leq \guess(X|\rE)/P[\Lambda]$ for any event $\Lambda$. 
\end{property}
Note that Property~\ref{prop:expansion} implies that $\guess(X|\rE,\Lambda) \leq \guess(X|1_\Lambda \rE)/P[\Lambda]$, but the statement of  Property~\ref{prop:cond} is stronger since 
$\guess(X|\rE) \leq \guess(X|1_\Lambda \rE)$. 


\begin{proof}[of Property~\ref{prop:cond}]
It holds that%
\footnote{Here and throughout, for operators $K$ and $L$, the inequality $K \leq L$ means that $L-K$ is positive-semidefinite. } 
 $P[\Lambda] \cdot \rho_{X \rE|\Lambda} \leq \rho_{X \rE}$, and hence that for any measurement ${\cal M}$ on $\rE$
$$
P[\Lambda] \cdot P[{\cal M}(\rE) \!=\! X|\Lambda] \leq P[{\cal M}(\rE) \!=\! X] \leq \guess(X|\rE) \enspace,
$$
which implies the claim. 
\qed
\end{proof}

\begin{property}\label{prop:bound}
There exists $\sigma_\rE \in \dens(\H_\rE)$ so that 
$$\rho_{X\rE} \leq \guess(X|\rE) \cdot \I_{\cal X} \otimes \sigma_\rE = \guess(X|\rE) \cdot |{\cal X}| \cdot \mu_{\cal X} \otimes \sigma_\rE \enspace.$$ 
\end{property}

\begin{proof}
The claim follows from Renner's original definition of the conditional min-entropy as
$$
H_\infty(X|\rE) := \max_{\sigma_\rE} \max_\lambda \Set{\lambda}{\rho_{X\rE} \leq 2^{-\lambda}\cdot \I_{\cal X} \otimes \sigma_\rE}
$$
and the identity $H_\infty(X|\rE) = -\log \guess(X|\rE)$, as shown in~\cite{KRS09}.  
\qed
\end{proof}

\section{Enabling Tools}

In this section, we introduce and discuss the main technical tools for the constructions
and analyses of our key-recycling authentication and encryption schemes.

\subsection{On Guessing the Outcome of Quantum Measurements}\label{sec:GuessingGames}

We consider different ``guessing games", where one or two players need to guess the outcome of a quantum measurement. The bounds are derived by means of the techniques of~\cite{TFKW13}.

\paragraph{\sc Two-player guessing. }

Here, we consider a game where two parties, Bob and Charlie, need to {\em simultaneously} and without communication guess the outcome of BB84 measurements performed by Alice (on $n$ qubits prepared by Bob and Charlie), when given the bases that Alice chose. 
This is very similar to the {\em monogamy game} introduced and studied in~\cite{TFKW13}, but in the version we consider here, the sequence of bases is not chosen from $\set{0,1}^n$ but from a code $\code \subset \set{0,1}^n$ with minimal distance $d$. It is useful to think of $d$ to be much larger than $\log|\code|$, i.e., the dimension of the code in case of a linear code. The following shows that in case of a uniformly random choice of the bases in $\code$, Bob and Charlie cannot do much better than to agree on a guess for the bases and to give Alice qubits in those bases. 

\begin{proposition}\label{prop:monogamy}
Let $\H_\rA$ be a $n$-qubit system, and let $\H_\rB$ and $\H_\rC$ be arbitrary quantum systems. Consider a state $\rho_{\Theta \rA\rB\rC} = \mu_{\code} \otimes \rho_{\rA\rB\rC} \in \dens(\code \otimes \H_\rA \otimes \H_\rB \otimes \H_\rC)$,
and let 
$$
\rho_{\Theta X X' X''} = {\cal N}_{\Theta \rC \to X''} \circ {\cal N}_{\Theta \rB \to X'} \circ {\cal M}^\BB_{\Theta \rA \to X}  \bigl(\rho_{\Theta\rA\rB\rC}\bigr)
$$
where ${\cal M}^\BB_{\Theta \rA \to X}$ is the BB84-measurement of the system $\rA$ (controlled by $\Theta$), and ${\cal N}_{\Theta \rB \to X'}$ and ${\cal N}_{\Theta \rC \to X''}$ are arbitrary (possibly different) measurements of the respective systems $\rB$ and $\rC$, both controlled by $\Theta$. Then, it holds that
$$
P[X' \!=\! X \wedge X'' \!=\! X] \leq \frac{1}{|\code|} + \frac{1}{2^{d/2}} \enspace.
$$
\end{proposition}

\begin{proof}
The proof uses the techniques from~\cite{TFKW13}. 
By Naimark's theorem, we may assume without loss of generality that the conditional measurements ${\cal N}_{\rB \to X'}^\theta$ and ${\cal N}_{\rC \to X''}^\theta$ are specified by families $\{P_x^{\theta}\}_{x}$ and $\{Q_x^{\theta}\}_x$ of {\em projections}. 
Then, defining for every $\theta \in \code$ the projection $\Pi^\theta = \sum_x H^\theta \proj{x}H^\theta \otimes P^\theta_x \otimes Q^\theta_x$, we see that
$$
P[X' \!=\! X \wedge X'' \!=\! X] \leq \frac{1}{|\code|} \biggl\| \sum_\theta \Pi^\theta \biggr\| \leq \frac{1}{|\code|} \sum_\delta \max_\theta \bigl\| \Pi^\theta \Pi^{\theta \oplus \delta} \bigr\| \enspace,
$$
where $\|\cdot\|$ refers to the standard operator norm, and the second inequality is by Lemma 2.2 in~\cite{TFKW13}. For any $\theta,\theta' \in \code$, bounding $\Pi^\theta$ and $\Pi^{\theta'}$ by 
$$
\Pi^\theta \leq \Gamma^\theta:= \sum_x H^\theta \proj{x}H^\theta \otimes P^\theta_x \otimes \I 
$$
and
$$
\Pi^{\theta'} \leq \Delta^{\theta'}:= \sum_x H^{\theta'} \proj{x}H^{\theta'} \otimes \I \otimes Q^{\theta'}_x \enspace,
$$ 
it is shown in~\cite{TFKW13} (in the proof of Theorem 3.4) that  
$$
\bigl\| \Pi^\theta \Pi^{\theta'} \bigr\| \leq \bigl\| \Gamma^\theta \Delta^{\theta'} \bigr\| \leq \frac{1}{2^{d_H(\theta,\theta')/2}} \leq \frac{1}{2^{d/2}}
$$ 
where the last inequality holds unless $\theta = \theta'$, from which the claim follows. 
\qed
\end{proof}

\begin{remark}
If we restrict $\H_\rB$ to be a $n$-qubit system too, and replace the (arbitrary) measurement ${\cal N}_{\Theta \rB}$ by a BB84 measurement ${\cal M}^\BB_{\Theta \rB}$, i.e., ``Bob measures correctly", then we get 
$$
P[X' \!=\! X \wedge X'' \!=\! X] \leq \frac{1}{|\code|} + \frac{1}{2^d} \enspace.
$$
\end{remark}

\paragraph{\sc Two-player guessing with quantum side information. }

Now, we consider a version of the game where Alice's choice for the bases is not uniformly random, and, additionally, Bob and Charlie may hold some quantum side information on Alice's choice at the time when they can prepare the initial state (for Alice, Bob and Charlie). 

\begin{corollary}\label{cor:monogamy}
Let $\H_\rA$ be a $n$-qubit system, and let $\H_\rB, \H_\rC$ and $\H_\rE$ be arbitrary quantum systems. Consider a state $\rho_{\Theta \rE} \in \dens(\code \otimes \H_\rE)$, and let 
$$
\rho_{\Theta\rA\rB\rC} = {\cal Q}_{\rE\to\rA\rB\rC} \bigl(\rho_{\Theta \rE}\bigr) \in \dens(\code \otimes \H_\rA \otimes \H_\rB \otimes \H_\rC)
$$
where ${\cal Q}_{\rE\to\rA\rB\rC}$ is a CPTP map acting on $\rE$ (only), 
and let 
$$
\rho_{\Theta X X' X''} = {\cal N}_{\Theta \rC \to X''} \circ {\cal N}_{\Theta \rB \to X'} \circ {\cal M}^\BB_{\Theta \rA \to X}  \bigl(\rho_{\Theta\rA\rB\rC}\bigr)
$$
as in Proposition~\ref{prop:monogamy} above. Then, it holds that
$$
P[X' \!=\! X \wedge X'' \!=\! X] \leq \guess(\Theta|\rE) + \frac{\guess(\Theta|\rE) \cdot |\code|}{2^{d/2}} \enspace.
$$
\end{corollary}

\begin{proof}
By Proposition~\ref{prop:monogamy}, the claim holds as $P[X' \!=\! X \wedge X'' \!=\! X] \leq 1/|\code| + 1/2^{d/2}$ for the special case where $\rho_{\Theta \rE}$ is of the form $\rho_{\Theta \rE} = \mu_\code \otimes \sigma_\rE$. Furthermore, by Property~\ref{prop:bound} we know that an arbitrary $\rho_{\Theta \rE}\in \dens(\code \otimes \H_\rE)$ is bounded by 
$$
\rho_{\Theta \rE} \leq \guess(\Theta|\rE) \cdot |\code| \cdot \mu_\code \otimes \sigma_\rE \enspace.
$$
Therefore, since the composed map
$$
\dens({\cal C} \otimes \H_\rE) \to \dens(\set{0,1}), \;
\rho_{\Theta \rE} \mapsto \rho_{\Theta\rA\rB\rC} \mapsto \rho_{X X' X''} \mapsto \rho_{1_{X=X' \wedge X=X''}}
$$
is still a CPTP map, it holds that for arbitrary $\rho_{\Theta \rE} \in \dens(\code \otimes \H_\rE)$
$$
P[X' \!=\! X \wedge X'' \!=\! X] \leq \guess(\Theta|\rE) \cdot |\code| \cdot \biggl(\frac{1}{|\code|} + \frac{1}{2^{d/2}} \biggr) \enspace,
$$
which proves the claim. 
\qed
\end{proof}

\begin{remark}
Similarly to the remark above, the bound relaxes to 
$$
P[X' \!=\! X \wedge X'' \!=\! X] \leq \guess(\Theta|\rE) + \frac{\guess(\Theta|\rE) \cdot |\code|}{2^d} \enspace,
$$
when ``Bob measure correctly". 
\end{remark}

\paragraph{\sc Single-player guessing (with quantum side information). }

Corollary~\ref{cor:monogamy} immediately gives us control also over a slightly different game, where only one party needs to guess Alice's measurement outcome, but here he is {\em not} given the bases. Indeed, any strategy here gives a strategy for the above simultaneous-guessing game, simply by ``pre-measuring" $\rB$, and having Bob and Charlie each keep a copy of the measurement outcome. 

\begin{corollary}\label{cor:guessing}
Let $\H_\rA$ be a $n$-qubit system, and let $\H_\rB$ and $\H_\rE$ be arbitrary quantum systems. Consider a state $\rho_{\Theta\rE}  \in \dens(\code \otimes \H_\rE)$ 
and let 
$$
\rho_{\Theta\rA\rB} = {\cal Q}_{\rE\to\rA\rB} \bigl(\rho_{\Theta \rE}\bigr) \in \dens(\code \otimes \H_\rA \otimes \H_\rB)
$$
where ${\cal Q}_{\rE\to\rA\rB}$ is a CPTP map acting on $\rE$, 
and let 
$$
\rho_{\Theta X X''} =  {\cal N}_{\rB \to X''} \circ {\cal M}^\BB_{\Theta \rA \to X} \bigl(\rho_{\Theta \rA\rB}\bigr)
$$
where ${\cal N}_{\rB \to X''}$ is an arbitrary measurement of $\rB$ (with no access to $\Theta$). Then, it holds that
$$
P[X'' \!=\! X] \leq \guess(\Theta|\rE) + \frac{\guess(\Theta|\rE) \cdot |\code|}{2^{d/2}} \enspace.
$$
In other words, for the state $\rho_{\Theta X \rB} = {\cal M}^\BB_{\Theta \rA \to X} (\rho_{\Theta \rA\rB})$ we have that
$$
\guess(X|\rB) \leq \guess(\Theta|\rE) + \frac{\guess(\Theta|\rE) \cdot |\code|}{2^{d/2}} \enspace.
$$
\end{corollary}

\begin{remark}
If we restrict the side information $\rE$ to be {\em classical} then, using slightly different techniques, we can improve the bounds from Corollary~\ref{cor:monogamy} and~\ref{cor:guessing} to
$$
\guess(\Theta|\rE) + \frac{1}{2^{d/2}} \enspace.
$$
Whether this improved bound also holds in case of quantum side information is an open question. 
\end{remark}

\subsection{Hash Functions with Message-Independence and Key-Privacy}\label{sec:key-privacy}

The goal of key-recycling is to be able to {\em re-use} a cryptographic key. For this to be possible, it is necessary\,---\,actually, not necessary but sufficient\,---\,that a key {\em stays secure}, i.e., that the primitive that uses the key does {\em not} reveal anything on the key, or only very little. We introduce here a general notion that captures this, i.e., that ensures that the key stays secure {\em as long as} there is {\em enough uncertainty} in the message the primitive is applied to\,---\,in our construction(s), this uncertainty will then be derived from the quantum part. 

Consider a keyed hash function $\hash: \K \times {\cal X} \to {\cal Y}$ with key space $\K$, message space $\cal X$, and range $\cal Y$. We define the following properties on such a hash function. 

\begin{definition}\label{def:uniformity}
We say that $\hash$ is {\em message-independent} if for a uniformly random key $K$ in $\K$, the distribution of the hash value $Y = \hash(K,x)$ is independent of the message $x \in \cal X$. 
And, we say that $\hash$ is {\em uniform} if it is message-independent and $Y = \hash(K,x)$ is uniformly random on $\cal Y$. 
\end{definition}
Thus, message-independence simply ensures that if the key is uniformly random and independent of the message, then the hash of the message is independent of the message too. The key-privacy property below on the other hand ensures that for any adversary that has arbitrary but {\em limited} information on the message and the hash value\,---\,but no direct information on the key\,---\,has (almost) no information on the key. 

\begin{definition}\label{def:key-privacy}
We say that $\hash$ offers {\em $\nu$-key-privacy} if for any state $\rho_{KXY\rE}$ in $\dens(\K \otimes {\cal X} \otimes {\cal Y} \otimes \H_\rE)$ with the properties that $\rho_{KX} = \mu_\K \otimes \rho_X$, $Y = \hash(K,X)$ and
$K\leftrightarrow XY\leftrightarrow \rE$,  
it holds that
$$
\delta(K,\UNIF_\K|Y\rE) \leq \frac{\nu}{2} \sqrt{\guess(X|Y\rE)\cdot|{\cal Y}|} \enspace.
$$
We say that $\hash$ offers {\em ideal key-privacy} if it offers $1$-key-privacy.
\end{definition}

\begin{remark}\label{rem:msgind}
Note that if $\hash$ is message-independent then for $X,Y$ and $\rE$ as above in Definition~\ref{def:key-privacy}, we have that $\guess(X|Y\rE) = \guess(X|\rE)$. 
\end{remark}

Not so surprisingly, the joint notion of uniformity and key-privacy is closely related to that of a {\em strong extractor}~\cite{NZ96}. Indeed, if $\hash$ is uniform and offers key-privacy then it is a strong extractor: for $\rho_{KXE} = \mu_\K \otimes \rho_{XE}$ and $Y = \hash(K,X)$, the condition on $\rho_{KXY\rE}$ in Definition~\ref{def:key-privacy} is satisfied, and thus we have the promised bound on $\delta(\rho_{KY\rE},\mu_\K \otimes \rho_{Y\rE}) = \delta(\rho_{KY\rE},\mu_\K \otimes \mu_{\cal Y} \otimes \rho_\rE)$, where the equality is due to uniformity. 
As such, \cite{RT00} shows that the required bound on $\delta(K,\UNIF_\K|Y\rE)$ is the best one can hope for. 
On the other hand, the following shows that from every strong extractor we can easily construct a hash function that offers uniformity and key-privacy. 

\begin{proposition}\label{prop:hashconstruction}
Let $\Ext: \K \times {\cal X} \to {\cal Y}$ be a strong extractor, i.e., for 
\mbox{$\rho_{K X \rE} = \mu_\K \otimes \rho_{X \rE} \in \dens(\K \otimes {\cal X} \otimes \H_\rE)$} and for $Y$ computed as $Y = \Ext(K,X)$ it holds that $\delta(\rho_{KY\rE},\mu_\K \otimes \rho_{Y\rE}) \leq \frac{\nu}{2} \sqrt{\guess(X|\rE)\cdot|{\cal Y}|}$. 
Furthermore, we assume that the range $\cal Y$ forms a group. Then, the keyed hash function%
\footnote{Here, and similarly in other occasions, $k\|k'$ is simply a synonym for the element $(k,k')$ in the Cartesian product of, here, $\K$ and ${\cal Y}$, and is mainly used to smoothen notation and avoid expressions like $\bigl((k,k'),x\bigr)$. }
$$
\hash: (\K \times {\cal Y}) \times {\cal X} \to {\cal Y}, \, (k\|k',x) \mapsto \Ext(k,x)+k' 
$$
with key space $\K \times {\cal Y}$ satisfies uniformity and $\nu$-key-privacy. 
\end{proposition}

\begin{proof} 
Uniformity is clear. For key-privacy, consider a state $\rho_{KXY\rE}$ with the properties as in Definition~\ref{def:key-privacy}. We fix an arbitrary $y \in \cal Y$ and condition on $Y = y$. Conditioning on $X = x$ as well for an arbitrary $x \in \cal X$, the key $(K,K')$ is uniformly distributed subject to $\hash(K,x)+K' = y$. In other words, $K$ is uniformly random in $\K$, and $K' = y - \hash(K,x)$. Therefore, making use of the Markov-chain property, conditioning on $Y = y$ only, $K$ is uniformly random in $\K$ and independent of $X$ and $\rE$, and $K' = y - \hash(K,X)$. 
Thus, by the extractor property, $\delta(\rho_{K'K\rE|Y=y},\mu_{\cal Y} \otimes \mu_\K \otimes \rho_{\rE|Y=y}) \leq \frac{\nu}{2} \sqrt{\guess(X|\rE,Y\!=\!y)\cdot|{\cal Y}|}$. The claim follows by averaging over $y$, and applying Jensen's inequality and Property~\ref{prop:expansion}. 
\qed
\end{proof}
The following technical result will be useful. 

\begin{lemma}\label{lemma:implchainrule}
Let $\hash: \K \times {\cal X} \to {\cal Y}$ be a keyed hash function that satisfies message-independence. Furthermore, let $\rho_{KXY\rE}$ be a state with the properties as in Definition~\ref{def:key-privacy}. Then
$$
\guess(X|KY\rE) \leq \guess(X|Y\rE) \cdot |{\cal Y}| \enspace.
$$
\end{lemma}

\begin{proof}
Note that the Markov-chain property $K\leftrightarrow XY\leftrightarrow \rE$ can be understood in that $\rE$ is obtained by acting on $X Y$ only: $\rE = {\cal Q}(X Y)$. For the purpose of the argument, we extend the state $\rho_{X K Y \rE}$ to a state $\rho_{X K K' Y Y' \rE \rE'}$ as follows. We choose a uniformly random and independent $K'$ in $\K$, and set $Y' = \hash(K',X)$ and $\rE' = {\cal Q}(X Y')$. Note that $\rho_{X K Y \rE}$ coincides with $\rho_{X K' Y' \rE'}$. Therefore, 
 $$
 \guess(X| Y \rE) =  \guess(X|Y' \rE') = \guess(X|K Y' \rE') \enspace,
 $$
 where the second equality is by the independence of $K$. 
 Furthermore, by Property~\ref{prop:cond}, we have that
 \begin{align*}
\guess(X|K Y' \rE') &\geq P[Y \!=\! Y'] \, \guess(X|K Y' \rE',Y \!=\! Y') \\
&= P[Y \!=\! Y'] \, \guess(X|K Y \rE,Y\!=\! Y') \enspace.
 \end{align*}
 Finally, by the message-independence of $\hash$, it holds that $Y'$ is independent of $K X Y \rE$ (and with the same distribution as $Y$), and therefore $P[Y \!=\! Y'] \geq 1/|{\cal Y}|$ and $\guess(X|K Y \rE,Y \!=\! Y') = \guess(X|K Y \rE)$. Altogether, this gives us the bound $\guess(X|K Y \rE) \leq  \guess(X| Y \rE) \cdot |{\cal Y}|$, which concludes the proof. 
\qed
\end{proof}
Equipped with Lemma~\ref{lemma:implchainrule}, we can now show the following composition results.

\begin{proposition}[Parallel Composition]\label{prop:parcomposition}
Consider two keyed hash functions $\hash_1: \K_1 \times {\cal X} \to {\cal Y}_1$ and $\hash_2: \K_2 \times {\cal X} \to {\cal Y}_2$ with the same message space $\cal X$, and 
 $$
 \hash: (\K_1 \times \K_2) \times {\cal X} \to {\cal Y}_1 \times {\cal Y}_2, \, (k_1\|k_2,x) \mapsto \bigl(\hash_1(k_1,x),\hash_2(k_2,x)\bigr) \, 
 $$
 with key space $\K = \K_1 \times \K_2$ and range ${\cal Y} = {\cal Y}_1 \times {\cal Y}_2$. 
 If $\hash_1$ and $\hash_2$ are both message-independent (or uniform) and respectively offer $\nu_1$-and $\nu_2$-key privacy, then $\hash$ is message-independent (or uniform) and offers $(\nu_1+\nu_2)$-key privacy. 
\end{proposition}
 
\begin{proof}
 That  message-independence/uniformity is preserved is clear. To argue key-privacy, assume that we have $\rho_{K_1 K_2 X} = \rho_{K_1} \otimes \rho_{K_2} \otimes \rho_X$, $Y_1 = \hash_1(K_1,X)$ and $Y_2 = \hash(K_2,X)$, and $K_1 K_2\leftrightarrow XY_1 Y_2\leftrightarrow \rE$. We need to bound the distance of $K_1 K_2$ from uniform when given $Y_1 Y_1\rE$, which we can decompose into 
$$
\delta(K_1 K_2,\UNIF_{\K_1} \UNIF_{\K_2} | Y_1 Y_1\rE) \leq \delta(K_1,\UNIF_{\K_1}|Y_1 Y_2 \rE) + \delta(K_2,\UNIF_{\K_2}|K_1 Y_1 Y_2 \rE) \enspace.
$$
The above conditions on $\rho_{K_1 K_2 X Y_1 Y_2 \rE}$ imply that $K_1 \leftrightarrow XY_1 \leftrightarrow K_2 Y_2 \rE$ holds, and thus also $K_1 \leftrightarrow XY_1 \leftrightarrow Y_2 \rE$. Indeed, $K_1 K_2\leftrightarrow XY_1 Y_2\leftrightarrow \rE$ implies that also $K_1 \leftrightarrow XY_1 K_2 Y_2\leftrightarrow \rE$, which together with $K_1 \leftrightarrow XY_1 \leftrightarrow K_2 Y_2$ (which holds by choice of $K_2$ and $Y_2$) implies that $K_1 \leftrightarrow XY_1 \leftrightarrow K_2 Y_2 \rE$. 
Therefore, by the key-privacy property of $\hash_1$, setting $\rE_1 = Y_2 \rE$, we see that 
$$
\delta(K_1,\UNIF_{\K_1}|Y_1 Y_2 \rE) \leq \frac{\nu_1}{2} \sqrt{\guess(X|Y_1 Y_2 \rE)\cdot|{\cal Y}_1|} \enspace.
$$
Similarly, $K_2 \leftrightarrow XY_2 \leftrightarrow K_1 Y_1 \rE$, and so by the key-privacy property of $\hash_2$, setting $\rE_2 = K_1 Y_1 \rE$, we conclude that 
\begin{align*}
\delta(K_2,\UNIF_{\K_2}|K_1 Y_1 Y_2 \rE) &\leq \frac{\nu_2}{2} \sqrt{\guess(X|Y_2 K_1 Y_1 \rE)\cdot|{\cal Y}_2|} \\
&\leq \frac{\nu_2}{2} \sqrt{\guess(X|Y_2 Y_1 \rE)\cdot|{\cal Y}_1|\cdot|{\cal Y}_2|}  \enspace,
\end{align*}
which proves the claim. 
\qed
\end{proof}

 \begin{proposition}[``Sequarallel'' Composition]\label{prop:seqcomposition}
Consider two keyed hash functions $\hash_1: \K_1 \times {\cal X} \to {\cal Y}_1$ and $\hash_2: \K_2 \times ({\cal X}\otimes{\cal Y}_1) \to {\cal Y}_2$ with message spaces as specified, and 
 \begin{align*}
 \hash: (\K_1 \times \K_2) \times {\cal X} \to {\cal Y}_1 \times {\cal Y}_2 \, , \;  (k_1\|k_2,x) \mapsto \bigl(\hash_1(k_1,x),\hash_2(k_2,x\|\hash_1(k_1,x))\bigr) 
 \end{align*}
 with key space $\K = \K_1 \times \K_2$ and range ${\cal Y} = {\cal Y}_1 \times {\cal Y}_2$. 
 If $\hash_1$ and $\hash_2$ are both message-independent (or uniform) and respectively offer $\nu_1$-and $\nu_2$-key privacy, then $\hash$ is message-independent (or uniform) and offers $(\nu_1+\nu_2)$-key privacy. 
 \end{proposition}

 \begin{proof}
The proof goes along the same lines as the proof of Proposition~\ref{prop:parcomposition}, except that in the reasoning for the bound on $\delta(K_2,\UNIF_{\K_2}| K_1 Y_1 \rE)$, we append $Y_1$ to $X$, with the consequence that we get a bound that is in terms of $\guess(X Y_1|Y_2 Y_1 \rE)$, but this obviously coincides with $\guess(X|Y_2 Y_1 \rE)$, and thus we end up with the same bound. 
 \qed
 \end{proof}

\section{Message Authentication with Key-Recycling}

We first consider the problem of {\em message authentication} with key-recycling. It turns out that\,---\,at least with our approach\,---\,this is the actual challenging problem, and extending to (authenticated) encryption is then quite easy. 

\subsection{The Semantics}

We quickly specify the semantics of a quantum authentication code (or scheme) with key-recycling.%
\footnote{Our definition is tailored to our goal that the key can be re-used unchanged in case the message is accepted by the recipient, Bob, and only needs to be refreshed in case he rejects. In the literature, key-recycling sometimes comes with {\em two} refresh procedures, one for the case Bob rejects and one for the case he accepts.  
}

\begin{definition}
A {\em quantum authentication code (with key recycling)} $\QMAC$ with message space $\MSG$ and key space $\KEY$ is made up of the following components:
 (1) A CPTP map $\Auth$ that is controlled by a {\em message} $\msg \in \MSG$ and a {\em key} $\key \in \KEY$, and that acts on an empty system and outputs a {\em quantum authentication tag} (with a fixed state space), (2) a measurement $\Verify$ that is controlled by $\msg \in \MSG$ and $\key \in \KEY$, and that acts on a quantum authentication tag and outputs a decision bit $d \in \set{0,1}$, and (3) a randomized function $\Refresh: \KEY \to \KEY$. 
\end{definition}

We will often identify an authentication code, formalized as above, with the obvious {\em authenticated-message-transmission protocol} $\pi_{\QMAC}(\msg)$, where Alice and Bob start with a shared key $\key \in \KEY$, and  Alice sends the message $\msg$ along with its quantum authentication tag prepared by means of $\Auth$ to Bob over a channel that is controlled by the adversary Eve, and, upon reception of the (possibly modified) message and tag, Bob verifies correctness using $\Verify$ and accordingly accepts or rejects. If he rejects, then Alice and Bob replace $\key$ by $\key' := \Refresh(\key)$.%
\footnote{Obviously, this requires Alice and Bob to exchange fresh randomness, i.e., the randomness for executing $\Refresh$, in a {\em reliable} and {\em private} way; how this is done is not relevant here. }
Note that, for any message $\msg \in \MSG$ and any strategy for Eve on how to interfere with the communication, the protocol $\pi_{\QMAC}(\msg)$ induces a CPTP map $\Exe[\pi_{\QMAC}(\msg)]: \dens(\KEY \otimes \H_\rE) \to \dens(\KEY \otimes \H_{\rE'})$ that describes the evolution of the shared key $\key$ and Eve's local system as a result of the execution of $\pi_{\QMAC}(\msg)$. 

Our goal will be to show that, for our construction given below, and for any behavior of Eve, the CPTP map $\Exe[\pi_{\QMAC}(\msg)]$ maps a key about which Eve has little information into a (possibly updated) key about which Eve still has little information\,---\,what it means here to ``have little information" needs to be specified, but it will in particular imply that it still allows Bob to detect a modification of the message. This then ensures re-usability of the quantum authentication code\,---\,with the same key as long as Bob accepts the incoming messages, and with the updated key in case he rejects.

\subsection{The Scheme}\label{sec:scheme}

Let $\MSG$ be an arbitrary non-empty finite set. We are going to construct a quantum message authentication code $\QMAC$ with message space $\MSG$. 
To this end, let $\MAC: \K \times (\MSG\times\set{0,1}^n) \to \TAG$ be a classical one-time message authentication code  with a message space $\MSG\times\set{0,1}^n$ for some $n \in \N$. We require $\MAC$ to be secure in the standard sense, meaning a modified message will be detected except with small probability $\eps_\MAC$. Additionally, we require $\MAC$ to satisfy message-independence and ideal key-privacy, as discussed in Sect.~\ref{sec:key-privacy}. Actually, it is sufficient if $\MAC(\,\cdot\,,\msg\|\,\cdot\,)$, i.e., the hash function $\K \times\set{0,1}^n \to \TAG$, $(k,x) \mapsto \MAC(k,\msg\|x)$ obtained by fixing $\msg$, satisfies message-independence and ideal key-privacy for any $\msg \in \MSG$. 
Assuming that $\MSG$ consists of bit strings of fixed size so that $\MSG\times\set{0,1}^n = \set{0,1}^N$ for some $N \in \N$, the canonical message authentication codes $\MAC: \bigl(\F_2^{\ell \times N} \times \F_2^\ell\bigr) \times \F_2^N \to \F_2^\ell$, $(A\|b,x) \mapsto Ax+b$ and $\MAC: \bigl(\F_{2^N} \times \F_2^\ell\bigr) \times \F_{2^N} \to \F_2^\ell$, $(a\|b,x) \mapsto trunc(a \cdot x)+b$, where $trunc: \F_{2^N} \to \F_2^\ell$ is an arbitrary surjective $\F_2$-linear map, are suitable choices; this follows directly from Proposition~\ref{prop:hashconstruction}. 
Finally, let $\code \subset \set{0,1}^n$ be a code with large minimal distance~$d$. 

Then, our quantum message authentication code $\QMAC$ has a key space $\KEY = \K \times \code$, where for a key $k \| \theta \in \K \times \code$ we refer to $k$ as the ``MAC key" and to $\theta$ as the ``basis key", and $\QMAC$ works as described in Figure~\ref{fig:qmac}

\begin{figure}[H]
\vspace{-2ex}
\begin{mybox}
\begin{description}\setlength{\parskip}{1ex}
\item{$\QMAC.\Auth(k\|\theta, \msg)$: } Choose a uniformly random $x \in \set{0,1}^n$ and output $n$ qubits $\rB_\circ$ in state $H^\theta \ket{x}$ and the classical tag $t = \MAC(k,\msg\|x)$. 
\item{$\QMAC.\Verify(k\|\theta, \msg,t)$: } Measure the qubits $\rB_\circ$ in bases $\theta$ to obtain $x'$ (supposed to~be~$x$), check that $t = \MAC(k,\msg\|x')$, and output $0$ or $1$ accordingly. 
\item{$\QMAC.\Refresh(k\|\theta)$: } Choose a uniformly random $\theta' \in \code$ and output $k\|\theta'$. 
\end{description}
\end{mybox}\vspace{-2.5ex}
\caption{The quantum message authentication code $\MAC$. }\label{fig:qmac}
\end{figure}

It is clear that as long as the MAC key $k$ is ``secure enough", the classical $\MAC$ takes care of an Eve that tries to modify the message $\msg$, and it ensures that such an attack is detected by Bob, except with small probability.  What is non-trivial to argue is that the MAC key (together with the basis key) indeed stays ``secure enough" over multiple executions of $\pi_\QMAC(\msg)$; this is what we show below.


\subsection{Analysis}

We consider an execution of the authenticated-message-transmission protocol $\pi_\QMAC(\msg)$ for a fixed message $\msg$.  
Let $\rho_{K \Theta \rE} \in \dens(\K \otimes \code \otimes \H_\rE)$ be the joint state {\em before} the execution, consisting of the MAC key $K$, the basis key $\Theta$, and Eve's local quantum system $\rE$. 
The joint state $\Exe[\pi_{\QMAC}(\msg)](\rho_{K\Theta \rE})$ {\em after} the execution is given by $\rho_{K \Theta' T D \rC} \in \dens(\K \otimes \code \otimes \TAG \otimes \set{0,1} \otimes \H_\rC)$, where $\Theta'$ is the (possibly) updated basis key, $T$ is the classical tag, $D$ is Bob's decision to accept or reject, and $\rC$ is Eve's new quantum system.
Eve's complete information after the execution of the scheme is thus given by $\rE' = T D \rC$. 

Recall that $T D \rC$ is obtained as follows from $K \Theta \rE$. Alice prepares BB84-qubits $\rB_\circ$ for uniformly random bits $X$ and with bases determined by $\Theta$, and she computes the tag $T := \MAC(K,\msg\|X)$. Then, Eve acts on $\rB_\circ \rE$ (in a way that may depend on $T$) and keeps one part, $\rC$, of the resulting state, and Bob measures the other part, $\rB$, to obtain $X'$ and checks with the (possibly modified) tag $T$ to decide on $D$. 

Note that by a standard reasoning, we can think of the BB84 qubits $\rB_\circ$ not as being prepared by first choosing the classical bits $X$ and then ``encoding" them into qubits with the prescribed bases $\Theta$, but by first preparing $n$ EPR pairs $\Phi^+_{\rA\rB_\circ}$ and then measuring the qubits in $\rA$ in the prescribed bases to obtain~$X$, i.e., $\rho_{K \Theta X \rB_\circ \rE} = {\cal M}^\BB_{\Theta \rA \to X} \bigl(\Phi^+_{\rA\rB_\circ} \otimes \rho_{K \Theta \rE}\bigr)$. 

\medskip

The following captures the main security property of the scheme. 

\begin{theorem}\label{thm:security}
If the state before the execution of $\pi_\QMAC(\msg)$ is of the form $\rho_{K\Theta \rE} = \mu_\K \otimes \rho_{\Theta \rE}$, then for any Eve the state $\rho_{K \Theta' \rE'} = \Exe[\pi_{\QMAC}(\msg)](\rho_{K\Theta \rE})$ after the execution satisfies 
$$
\guess(\Theta'|\rE') \leq  \guess(\Theta|\rE) + \frac{1}{|\code|}
$$
and
$$
\delta(K, \UNIF_\K |\Theta' \rE') \leq 2\eps_\MAC + \frac{\sqrt2}{2}\sqrt{\guess(\Theta|\rE) \biggl( 1+ \frac{|\code|}{2^{d/2}}\biggr) \, |\TAG|}  \enspace.
$$
\end{theorem}
This means that if {\em before} the execution of $\pi_\QMAC(\msg)$, it holds that Eve's guessing probability on $\Theta$ is small and $K$ looks perfectly random to her (even when given $\Theta$), then {\em after} the execution, Eve's guessing probability on (the possibly refreshed) $\Theta'$ is still small and $K$ still looks {\em almost} perfectly random to her. As such, we may then consider a {\em hypothetical} refreshing of $K$ that has almost no impact, but which brings us back to the position to apply Theorem~\ref{thm:security} again, and hence allows us to re-apply this ``preservation of security'' for the next execution, and so on. This in particular allows us to conclude that in an arbitrary {\em sequence} of executions, the MAC key $K$ stays almost perfectly random for Eve, and thus any tampering with an authenticated message will be detected by Bob except with small probability by the security of $\MAC$ (see Sect.~\ref{sec:reuse} for more details). 

\begin{remark}
For simplicity, in Theorem~\ref{thm:security} and in the remainder of this work, we assume the message $\msg$ to be arbitrary but {\em fixed}. However, it is not hard to see that we may also allow $\msg$ to be obtained by means of a measurement, applied to Eve's system $\rE$ before the execution of $\pi_\QMAC(\msg)$, i.e., Eve can choose it. The bounds of Theorem~\ref{thm:security} then hold {\em on average over the measured $\msg$}. This follows directly from Property~\ref{prop:expansion} for the bound the the guessing probability, and from a similar decomposition property for the trace distance, together with Jensen's inequality for the bound on the trace distance. 
We emphasize however, that we do assume $\msg$, even when provided by Eve, to be {\em classical}, i.e., we do not consider so-called superposition attacks. 
\end{remark}

The formal proof of Theorem~\ref{thm:security} is given below; the intuition is as follows. For the bound on the guessing probability of the (possibly updated) basis key, we have that in case Bob rejects and so the basis key is re-sampled from $\code$, Eve has obviously guessing probability $1/|\code|$. In case Bob accepts, the fact that Bob accepts may increase Eve's guessing probability. For instance, Eve may measure one qubit in, say, the computational basis, and forward the correspondingly collapsed qubit to Bob; if Bob then accepts it is more likely that this qubit had been prepared in the computational basis by Alice, giving Eve some (new) information on the basis key. However, the resulting increase in guessing probability is {\em inverse proportional} to the probability that Bob actually accepts, so that this advantage is ``canceled~out" by the possibility that Bob will not accept. 
For the bound on the ``freshness'' of $K$ (given the basis key~$\Theta'$), by key privacy it is sufficient to argue that Eve has small guessing probability for $X$. In case Bob rejects, the (refreshed) basis key is useless to her for guessing $X$, and so the task of guessing $X$ reduces to winning the game considered in Corollary~\ref{cor:guessing}. Similarly, the case where Bob accepts fits into the game in Corollary~\ref{cor:monogamy}. In both cases, we get that the guessing probability of $X$ essentially coincides with $\guess(\Theta|\rE)$.

\begin{proof}
For the first claim, we simply observe that 
\begin{align*}
\guess(\Theta'|T D \rC) 
&=  \sum_{d=0}^1 P_{D}(d) \, \guess(\Theta'|T \rC, D\!=\!d) && \text{(by Property~\ref{prop:expansion})} \\
&=  P_{D}(0) \frac{1}{|\code|}  + P_{D}(1) \, \guess(\Theta|T \rC,D\!=\!1)\hspace{-5ex} \\
&\leq  P_{D}(0) \frac{1}{|\code|}  +  \guess(\Theta|T \rC)  && \text{(by Property~\ref{prop:cond})} \\
&\leq \frac{1}{|\code|} + \guess(\Theta|T \rB_\circ \rE) && \text{(by Property~\ref{prop:monotonicity})} \\
&= \frac{1}{|\code|} + \guess(\Theta| \rB_\circ \rE) && \text{(by Definition~\ref{def:uniformity})}  \\
&= \frac{1}{|\code|} + \guess(\Theta|\rE) \enspace.
\end{align*}
where the second equality holds because $\Theta'$ is freshly chosen in case Bob rejects and $\Theta' = \Theta$ in case he accepts, and the final equality holds because of the fact that $\rho_{\rB_\circ \Theta \rE} = \tr_X \circ {\cal M}^\BB_{\Theta \rA \to X} \bigl(\Phi^+_{\rA\rB_\circ} \otimes \rho_{\Theta \rE}\bigr) = \tr_\rA \bigl(\Phi^+_{\rA\rB_\circ}\bigr) \otimes \rho_{\Theta \rE} = \mu_{\rB_\circ} \otimes \rho_{\Theta \rE}$. 

\smallskip
 
For the second claim, consider $\tilde{D}$ and $\tilde\Theta'$ as follows. $\tilde{D}$ is $1$ if $X=X'$ and Eve has not modified the tag $T$ nor the message $\msg$, and $\tilde{D}$ is $0$ otherwise (i.e., $\tilde{D}$ is an ``ideal version'' of Bob's decision), and $\tilde\Theta'$ is freshly chosen if and only $\tilde{D} = 0$. The states of $K\Theta'TD \rC$ and $K\tilde\Theta'T\tilde{D} \rC$ are identical except for when $D = 1$ but $X \neq X'$ or Eve has modified $T$ or $\msg$, which happens with probability at most $\eps_\MAC$ by the security of $\MAC$, and thus the two states are $\eps_\MAC$-close. 
Therefore, $\delta(K, \UNIF_\K |\Theta'T D \rC) \leq  \delta(K, \UNIF_\K |\tilde\Theta'T\tilde{D} \rC) + 2 \eps_\MAC$, and so it suffices to analyze the state of $K\tilde\Theta'T\tilde{D} \rC$. 
Furthermore, we may assume that Eve's state $\rC$ contains the information of whether she modified $T$ or $\msg$, so that $\tilde{D}$ can be computed from $1_{X=X'}$ when given $\rC$, and thus $\delta(K,U_\K|\tilde\Theta'\tilde{D} T \rC) \leq \delta(K,U_\K|\tilde\Theta'\,1_{X=X'} \tilde{D} T \rC) = \delta(K,U_\K|\tilde\Theta'\,1_{X=X'} T \rC)$. 


Now, since $K$ is random and independent of $X \Theta \rB_\circ \rE$, $T$ is computed as $T = \MAC(K,\msg|X)$, and $\tilde\Theta \,1_{X=X'} \rC$ is obtained by acting on $T$ and $X \Theta \rB_\circ \rE$ only (and not on $K$), we see that the conditions required in Definition~\ref{def:key-privacy} are satisfied. 
Therefore, by the key-privacy of $\MAC$, and recalling Remark~\ref{rem:msgind}, 
$$
\delta(K,U_\K|T \tilde\Theta'\,1_{X=X'} \rC)  \leq \frac12 \sqrt{ \guess(X|\tilde\Theta' \,1_{X=X'} \rC) \,|\TAG| } \enspace.
$$
Furthermore, by Property~\ref{prop:expansion}, and noting that $\tilde\Theta'$ is freshly chosen whenever $X\!\neq\!X'$, and, depending on~$\rC$, freshly chosen or equal to $\Theta$ otherwise, 
\begin{align*}
\guess(X|\tilde\Theta' \,1_{X=X'} \rC) 
&= P[X\!\neq\!X'] \, \guess(X|\tilde\Theta'\rC,X\!\neq\!X') 
+ P[X\!=\!X'] \, \guess(X|\tilde\Theta' \rC,X\!=\!X') \\
&\leq P[X\!\neq\!X'] \, \guess(X|\rC,X\!\neq\!X') 
+ P[X\!=\!X'] \, \guess(X|\Theta \rC,X\!=\!X') \enspace.
\end{align*}
For the first term, we see that 
\begin{align*}
P[X\!\neq\!X'] \, \guess(X|\rC,X\!\neq\!X') 
&\leq \guess(X|\rC)  && \text{(by Property~\ref{prop:cond})} \\
&\leq\guess(X|T\rB_\circ\rE)   && \text{(by Property~\ref{prop:monotonicity})} \\
&\leq \guess(X|\rB_\circ\rE) && \text{(by Definition~\ref{def:uniformity})}  \\[-0.5ex]
&\leq \guess(\Theta|\rE)\bigl(1 +\textstyle\frac{|\code|}{2^{d/2}}\bigr
)  \enspace, \hspace{-3ex}
\end{align*}
where the final inequality follows from Corollary~\ref{cor:guessing} by recalling that $\rho_{\Theta X \rB_\circ \rE} = {\cal M}^\BB_{\Theta \rA \to X} \bigl(\Phi^+_{\rA\rB_\circ} \otimes \rho_{\Theta \rE}\bigr)$. Similarly, writing $X''$ for the measurement outcome when measuring $\rC$ using an optimal measurement ${\cal N}_{\Theta \rC}$ (controlled by $\Theta$), we obtain 
\begin{align*}
P[X\!=\!X'] \, \guess(X|\Theta \rC,X\!=\!X')  
&\leq P[X\!=\!X'] \, P[X \!=\! X''|X\!=\!X']  \\[0.5ex]
&\leq P[X\!=\!X' \wedge X \!=\! X'']  \\
&\leq \guess(\Theta|\rE)\bigl(1 +\textstyle\frac{|\code|}{2^{d/2}}\bigr
) \enspace,
\end{align*}
where the final inequality follows from Corollary~\ref{cor:monogamy} by observing that, using uniformity of $\MAC$ (Definition~\ref{def:uniformity}) in the second equality,  
\begin{align*}
\rho_{\Theta X X X''} &= {\cal N}_{\Theta \rC \to X''} \circ {\cal M}^\BB_{\Theta \rB \to X'} \circ {\cal Q}_{T \rB_\circ \rE \to \rB \rC} \bigl( \rho_{\Theta X T \rB_\circ \rE} \bigr) \\
&= {\cal N}_{\Theta \rC \to X''} \circ {\cal M}^\BB_{\Theta \rB \to X} \circ {\cal Q}_{T \rB_\circ \rE \to \rB \rC} \bigl( \rho_{\Theta X \rB_\circ \rE} \otimes \rho_T \bigr) 
\\
&= {\cal N}_{\Theta \rC \to X''} \circ {\cal M}^\BB_{\Theta \rB \to X'} \circ {\cal Q}_{T \rB_\circ \rE \to \rB \rC} \circ {\cal M}^\BB_{\Theta \rA \to X} \bigl(\Phi^+_{\rA\rB_\circ} \otimes \rho_{\Theta \rE} \otimes \rho_T \bigr) \\
&= {\cal N}_{\Theta \rC \to X''} \circ {\cal M}^\BB_{\Theta \rB \to X'} \circ {\cal M}^\BB_{\Theta \rA \to X} \circ {\cal Q}_{T \rB_\circ \rE \to \rB \rC} \bigl(\Phi^+_{\rA\rB_\circ} \otimes \rho_{\Theta \rE} \otimes \rho_T \bigr) \\
&= {\cal N}_{\Theta \rC \to X''} \circ {\cal M}^\BB_{\Theta \rB \to X'} \circ {\cal M}^\BB_{\Theta \rA \to X} \circ {\cal Q}'_{\rE \to \rA\rB \rC} \bigl(\rho_{\Theta \rE}\bigr) 
\end{align*} 
where ${\cal Q}'_{\rE \to \rA\rB \rC}$ is the CPTP map ${\cal Q}'_{\rE \to \rA\rB \rC}(\sigma_\rE) = {\cal Q}_{T \rB_\circ \rE \to \rB \rC}(\Phi^+_{\rA\rB_\circ} \otimes \sigma_\rE \otimes \rho_T)$. 
Collecting the terms gives the claimed bound. 
\qed
\end{proof}

\subsection{Re-usability of $\QMAC$}\label{sec:reuse}

We formally argue here that Theorem~\ref{thm:security}, which analyses a {\em single} usage of $\QMAC$, implies {\em re-usability}. The reason why this is not completely trivial is that after one execution of $\pi_\QMAC$, the MAC key $K$ is not perfectly secure anymore but ``only'' almost-perfectly secure, so that Theorem~\ref{thm:security} cannot be directly applied anymore for a second execution. However, taking care of this is quite straightforward. 

Formally, we have the following result regarding the re-usability of $\QMAC$. 

\begin{proposition}\label{prop:seqex}
If Alice and Bob start off with a uniformly random key, then for a sequence $\pi_{\QMAC}(\msg_1)$, $\pi_{\QMAC}(\msg_2),\ldots$ of sequential executions of protocol $\pi_{\QMAC}$, and for any strategy for Eve and any $i \in \N$, the probability $\eps_i$ that Eve {\em modifies} $\msg_{i}$ in the execution of $\pi_{\QMAC}(\msg_{i})$ yet Bob {\em accepts} is bounded by
$$
\eps_i \leq (2i-1) \cdot \eps_\MAC +  \frac{\sqrt2}{2} \sum_{j < i} \sqrt{\frac{j}{|{\code}|} \biggl( 1+ \frac{|\code|}{2^{d/2}}\biggr) \, |\TAG|} \enspace.
$$
\end{proposition}

\begin{proof}
In case $i = 1$, the statement reduces to $\eps_i \leq \eps_\MAC$, which holds by construction of $\QMAC$. To argue the general case, let $\rho_{K_1 \Theta_1 \rE_1},\rho_{K_2 \Theta_2 \rE_2},\ldots$ describe the evolution of the MAC key and the basis key and Eve's information on them, given that we start with a perfect key $\rho_{K_0 \Theta_0 \rE_0} =  \mu_{\K} \otimes \mu_{\code} \otimes\rho_{\rE_0}$. Formally, $\rho_{K_i \Theta_i \rE_i}$ is inductively defined as $\rho_{K_i \Theta_i \rE_i} = \Exe[\pi_{\QMAC}(\msg_i)](\rho_{K_{i-1} \Theta_{i-1} \rE_{i-1}})$. 
For the sake of the argument, we also consider $\tilde\rho_{K_1 \Theta_1 \rE_1},\tilde\rho_{K_2 \Theta_2 \rE_2},\ldots$ obtained by means of setting $\tilde\rho_{K_0 \Theta_0 \rE_0} = \rho_{K_0 \Theta_0 \rE_0}$ and $\tilde\rho_{K_i \Theta_i \rE_i} = \Exe[\pi_{\QMAC}(\msg_i)](\mu_{\K} \otimes \tilde\rho_{\Theta_{i-1} \rE_{i-1}})$, i.e., the evolution of the keys and Eve's information in a hypothetical setting where the MAC key is refreshed before every new execution. For these latter states $\tilde\rho_{K_i \Theta_i \rE_i} $, we can inductively apply Theorem~\ref{thm:security} and conclude that 
$$
\guess(\Theta_i|\rE_i) \leq \frac{i+1}{|{\code}|} \vspace{-1.7ex}
$$
and
$$
\delta(\tilde\rho_{K_i \Theta_i \rE_i},\mu_\K \otimes \tilde\rho_{\Theta_i \rE_i}) \leq \delta_i := 2\eps_\MAC + \frac{\sqrt2}{2}\sqrt{\frac{i}{|{\code}|} \biggl( 1+ \frac{|\code|}{2^{d/2}}\biggr) \, |\TAG|}  
$$
for any $i \in \N$. 
But now, for the original states $\rho_{K_1 \Theta_1 \rE_1},\rho_{K_2 \Theta_2 \rE_2},\ldots$, from the triangle inequality  we obtain that
\begin{align*}
\delta(\rho_{K_i \Theta_i \rE_i},\mu_\K \otimes \tilde\rho_{\Theta_i \rE_i}) &\leq 
\delta(\rho_{K_i \Theta_i \rE_i},\tilde\rho_{K_i \Theta_i \rE_i}) + \delta(\tilde\rho_{K_i \Theta_i \rE_i},\mu_\K \otimes \tilde\rho_{\Theta_i \rE_i}) \\[0.7ex]
&\leq \delta(\rho_{K_{i-1} \Theta_{i-1} \rE_{i-1}},\mu_\K \otimes \tilde\rho_{\Theta_{i-1} \rE_{i-1}}) + \delta_i \\
&\leq \sum_{j \leq i} \delta_j \enspace,
\end{align*}
where the last inequality is by induction (where the base case $i = 0$ is trivially satisfied). 
It now follows by basic properties of the trace distance that we have $\eps_{i+1} \leq \eps_\MAC + \sum_{j \leq i} \delta_j$. This proves the claim. 
\qed
\end{proof}

\subsection{Choosing the Parameters}\label{sec:pars}

Let $\lambda \in \N$ be the security parameter. 
Consider a $\MAC$ with $\eps_\MAC = 2^{-\lambda}$ and $|\TAG| = 2^\lambda$. This can for instance be achieved with the constructions suggested in Sect.~\ref{sec:scheme}. Also, consider a code $\code$ with $|\code| = 2^{3\lambda}$ and $d = 6\lambda$, so that $|\code|/2^{d/2} = 1$. The description of the basis key $\theta$ thus requires $3\lambda$ bits, and, by Singleton bound, it is necessary that $n \geq 9\lambda-1$.
Then, the bound in Proposition~\ref{prop:seqex} becomes
$$
\eps_{i+1} \leq (2i+1) \cdot 2^{-\lambda} + \sum_{j \leq i} \frac{\sqrt2}{2}\sqrt{\frac{j}{2^{3\lambda}} \biggl( 1+ \frac{2^{3\lambda}}{2^{3\lambda}}\biggr) \,2^\lambda} = \Big(2i + 1 + \sum_{j \leq i}\sqrt{j}\Big) 2^{-\lambda} \enspace.
$$
Hence, the error probability increases at most as $\big(i \sqrt{i} + 2i + 1\big)  2^{-\lambda}$ with the number $i$ of executions.

\subsection{Impersonation Attacks}

The above analysis considers a {\em substitution attack}, where first Alice is activated and is given a message $\msg$ (possibly produced by Eve) and she computes and sends the quantum authentication tag, and then Bob decides to accept or reject depending on the received (and possibly modified) tag. It is quite clear that our results carry over to an {\em impersonation attack}, where Alice is inactive and Eve tries to get a message accepted by Bob. Indeed, in this case, Bob will reject almost with certainty and $k$ remains close to uniformly random (while $\theta$ needs to be refreshed). 

However, a subtle issue in this context is the following.%
\footnote{This issue was pointed out to us by Christopher Portmann. }
Assume that after such an impersonation attack, Bob uses the key $k$ (which is still ensured to be secure) in another application, which may leak information on $k$ to Eve (like as one-time-pad encryption key on a message that is partly known to Eve). If Alice, who may be unaware of what has been going on, now gets activated on a message $\msg$ that Eve chose {\em dependent} on the information she just got on $k$ (and possibly by performing a measurement on her quantum state) and computes and sends a quantum authentication tag for $\msg$ using the key pair $(\theta,k)$, then this may potentially leak {\em additional} information on $k$ to Eve (which could then allow Eve to learn more information on the one-time-pad encrypted message). Indeed, this kind of attack is not covered by our analysis because our proof assumed $\msg$ to be {\em independent} of $k$. 

An easy way to overcome this issue is to enforce $\MAC$ to be of the form $\MAC(A||b,m) = A m + b$, where $A$ is a random (possibly Toeplitz-like) matrix; this was anyway our suggestion for $\MAC$. 
Then, we can write the tag $t$ as
$$
   t = \MAC(A||b,\msg||x) = A \, \Big[\begin{array}{c} \msg \\[-0.6ex] x \end{array} \Big] + b = A\, \Big[\begin{array}{c} \msg \\[-0.6ex] 0 \end{array} \Big] + A\, \Big[\begin{array}{c} 0 \\[-0.6ex] x \end{array} \Big] + b \, ,
$$
and since $x$ is proven to have high min-entropy (and is independent of $A$, unlike $\msg$), this acts as a strong extractor and thus $A \big[\begin{smallmatrix}0\\x\end{smallmatrix}\big]$ is close to uniformly random, given $k = A||b$ and the BB84 qubits $B_\circ$ and Eve's quantum state {\em before} she has done the measurement to obtain $\msg$, and thus also $t$ is close to uniformly random, given $k$ and the BB84 qubits and the Eve's quantum state {\em after} she has done the measurement. As such, the quantum authentication tag does {\em not} leak any (additional) information on $k$ to Eve in the above kind of attack.

\section{Extensions and Variations}

We show how to modify our scheme $\QMAC$ as to offer encryption as well, i.e., to produce an authenticated encryption of $\msg$, and how to deal with noise in the quantum communication; we start with the latter since this is more cumbersome. At the end of the section, we show how to tweak our schemes so as to be able to authenticate and/or encrypt {\em quantum} messages as well, and we discuss some variations.

\subsection{Dealing with Noise}

In order to deal with noise in the quantum communication, we introduce the following primitive. We consider a keyed hash function $\SS: {\cal L} \times \set{0,1}^n \to {\cal S}$ that has the property that given the key $\ell$, the {\em secure sketch} $s = \SS(\ell,x)$ of the message $x$, and a ``noisy version'' $x'$ of $x$, i.e., such that $d_H(x,x') \leq \noise\cdot n$ for some given noise parameter $\varphi < \frac12$, the original message $x$ can be recovered except with probability $\eps_\SS$. Additionally, we want $\SS$ to satisfy the message-independence and ideal key-privacy properties from Definitions~\ref{def:uniformity} and~\ref{def:key-privacy}. 
Such constructions exist for small enough $\phi > 0$, as discussed in Appendix~\ref{sec:SSConstruction}, based on techniques by Dodis and Smith~\cite{DS05}. 

Then, the key for our noise-tolerant quantum message authentication code $\QMAC^*$ consists of a (initially) uniformly random MAC key $k \in \K$ for $\MAC$, an (initially) uniformly random secure-sketch key $\ell \in \cal L$ for $\SS$, and an (initially) random and independent basis key $\theta$, chosen from the code $\code \subset \set{0,1}^n$, and the scheme works as described in Figure~\ref{fig:qmacx}.

\begin{figure}
\begin{mybox}
\begin{description}\setlength{\parskip}{1ex}
\item{$\QMAC^*\!.\Auth(k\|\ell\|\theta, \msg)$: } Choose a uniformly random $x \in \set{0,1}^n$ and output $n$ qubits $\rB_\circ$ in state $H^\theta \ket{x}$ together with the secure sketch $s = \SS(\ell,x)$ and the tag $t = \MAC(k,\msg\|x\|s)$. 
\item{$\QMAC^*\!.\Verify(k\|\ell\|\theta, \msg,t)$: } Measure the qubits $\rB_\circ$ in bases $\theta$ to obtain~$x'$, recover (what is supposed to be) $x$ using the secure sketch $s$, and check the tag~$t$. If this check fails or $d_H(x,x') > \noise\cdot n$ then output $0$, else~$1$. 
\item{$\QMAC^*\!.\Refresh(k\|\ell\|\theta)$: } Choose uniformly random $\theta' \in \code$ and output $k\|\ell\|\theta'$. 
\end{description}
\end{mybox}\vspace{-2.5ex}
\caption{The noise-tolerant quantum message authentication code $\MAC^*$. }\label{fig:qmacx}
\vspace{-2ex}
\end{figure}


\begin{theorem}\label{thm:securitywithnoise}
If the state before the execution of $\pi_{\QMAC^*}(\msg)$ is of the form $\rho_{K\Theta \rE} = \mu_\K \otimes \rho_{\Theta \rE}$, then for any Eve the state $\rho_{K \Theta' \rE'} = \Exe[\pi_{\QMAC^*}(\msg)](\rho_{K\Theta \rE})$ after the execution satisfies 
$$
\guess(\Theta'|\rE') \leq  \guess(\Theta|\rE) + \frac{1}{|\code|}
$$
and
$$
\delta(KL,U_{\K \times {\cal L}} |\Theta' \rE') 
%
\leq  2(\eps_\MAC + \eps_\SS) + \sqrt{\guess(\Theta|\rE) \biggl( 2\!+\! \frac{|\code|}{2^{d/2}} \!+\! \frac{|\code|\cdot 2^{h(\noise)n}}{2^{d}} \biggr) |\TAG| |{\cal S}|  } \enspace .
$$
\end{theorem}

\begin{proof}
The proof of the first statement, i.e., the bound on $\guess(\Theta'|\rE')$ is exactly like in the proof of Theorem~\ref{thm:security}, with the only exception that in the one expression where the tag $T$ appears (i.e. in the expression obtained by using Property~\ref{prop:monotonicity}), now $S$ appears as well (along with $T$); but like $T$, it disappears again in the next step due to message-independence.

For the bound on $\delta(KL,U_{\K \times {\cal L}} |\Theta' \rE')$ we follow closely the proof of Theorem~\ref{thm:security} but with the following modifications. 

\begin{enumerate}\setlength{\parskip}{1ex}
\item The key $K$ is replaced by the key pair $(K,L)$, and the tag $T$ by the tag-secure-sketch pair $(T,S)$, and we observe that we can understand $(T,S)$ to be the hash of the input $X$ under key $(K,L)$ with respect to a hash function that satisfies message-independent and (almost) key-privacy. Indeed, this composed hash function can be understood as being obtained by means of Proposition~\ref{prop:seqcomposition}. As such, whenever we argue by means of message-independence (Definition~\ref{def:uniformity}) or key-privacy (Definition~\ref{def:key-privacy}) in the proof of Theorem~\ref{thm:security}, we can still do so, except that we need to adjust the bound on the uniformity of the key to the new\,---\,and now composite\,---\,hash function. 
\item The auxiliary random variable $\tilde{D}$, and correspondingly $\tilde\Theta'$, is defined in a slightly different way: $\tilde{D}$ is $1$ if $X \epsclose[\noise] X'$ and Eve has not modified the tag $T$, the secure-sketch $S$, nor the message $\msg$. The ``real'' state with $D$ and $\Theta'$ is then $(\eps_\MAC + \eps_\SS)$-close to the modified one with $\tilde{D}$ and $\tilde\Theta'$ instead. Correspondingly, the decomposition of the distance to be bounded is then done with respect to the indicated random variable $1_{X \epsclose[\noise] X'}$ instead of $1_{X = X'}$. 
\item When bounding the probability $P[X \epsclose[\noise]\! X' \wedge\, X \!=\! X'']$, we refer to the game analyzed in Corollary~\ref{cor:monogamy+} in Appendix~\ref{app:monogamy+}, which applies to the situation here where some slack is given for Bob's guess. 
\end{enumerate}
The claimed bound is then obtained by adjusting terms according to the above changes: update the bounds obtained by applying Definition~\ref{def:key-privacy} to the updated bound $\sqrt{\guess(X|\cdots) \, |\TAG| \, |{\cal S}|}$, obtained by means of Proposition~\ref{prop:seqcomposition}, and inserting the $2^{h(\noise)n}$ blow-up when using Corollary~\ref{cor:monogamy+} instead of Corollary~\ref{cor:monogamy}, but making use of the observation in Remark~\ref{rem:honestBob}. 
\qed
\end{proof}

In essence, compared to the case with no noise, we have an additional loss due to the $|{\cal S}|$ term, whereas we can neglect the term with $2^{h(\noise)n}$ for small enough~$\varphi$. To compensate for this additional loss, we need to have $\varsigma = \log|{\cal S}|$ additional bits of entropy in $\Theta$, i.e., we need to choose $\code$ with $|\code| = 2^{3\lambda+\varsigma}$ and $d = 6\lambda+2\varsigma$. By Singleton bound, this requires $n \geq 9\lambda + 3\varsigma - 1$, and thus puts a bound $\varsigma < n/3$ on the size of the secure sketch, and thus limits the noise parameter~$\phi$.%
\footnote{We recall that, when using a $\delta$-biased family of codes to construct the secure sketch $\SS$, as discussed in the Appendix~\ref{sec:SSConstruction}, then $\varsigma$ does not correspond exactly to the size of the syndrome given by the code, but is determined by the parameter $\delta$, and is actually somewhat larger than the size of the syndrome. }

\subsection{Adding Encryption}\label{sec:Enc}

Adding encryption now works pretty straightforwardly. Concretely, our quantum encryption scheme with key recycling $\QENC^*$ is obtained by means of the following modifications to $\QMAC^*$. 
Alice and Bob extract additional randomness from $x$ using an extractor that offers message-independence and key-privacy, and use the extracted randomness as one-time-pad key to en-/decrypt $\msg$. Finally, the resulting ciphertext $c$ is authenticated along with $x$ and~$s$; this is in order to offer authenticity as well and can be omitted if privacy is the only concern. 

Security can be proven along the same lines as Theorem~\ref{thm:security}, respectively Theorem~\ref{thm:securitywithnoise} for the noise-tolerant version, and Proposition~\ref{prop:seqex}: we simply observe by means of Proposition~\ref{prop:parcomposition} and~\ref{prop:seqcomposition} that the composition of computing the triple $c$, $s$ and $t = \MAC(k,x\|c\|s)$ from $x$ constitutes a keyed hash function that offers message-independence and key-privacy, and then we can argue exactly as above to show that the (possibly refreshed) key stays secure over many executions. Also, given that the key is secure before an execution, we can control the min-entropy in $X$ as in the proof of Theorem~\ref{thm:securitywithnoise} and argue almost-perfect security of the extracted one-time-pad key, implying privacy of the communicated message. 

In order to accommodate for the additional entropy that is necessary to extract this one-time-pad key, which is reflected in the adjusted range of the composed keyed hash function, we now have to choose $\code$ with $|\code| = 2^{3\lambda+\varsigma+m}$ and $d = 6\lambda+2\varsigma+2m$, where $m = \log|\MSG|$; this requires $n \geq 9\lambda + 3\varsigma + 3m - 1$ by Singleton bound.

\subsection{Optimality of the Key Recycling}

Our aim was, like in~\cite{DPS05,DPS14}, to minimize the number of fresh random bits needed for the key refreshing. In our constructions, where the key is refreshed simply by choosing a new basis key $\theta$, this number is obviously given by the number of bits needed to represent $\theta$, i.e., in the above encryption scheme $\QENC^*$, it is
$$
\log|\code| = 3\lambda+\varsigma +m \enspace. 
$$
This is close to optimal for large messages and assuming almost no noise, so that $m \gg \lambda,\varsigma$. 
Indeed, assuming that Eve knows the encrypted message, i.e., we consider a known-plaintext attack, it is not hard to see that for any scheme that offers (almost) perfect privacy of the message, by simply keeping everything that is communicated from Alice to Bob, in particular by keeping all qubits that Alice communicates (which will most likely trigger Bob to reject), Eve can always learn (almost) $m$ bits of Shannon information on the key. As such, it is obviously necessary that the key is updated with (almost) $m$ fresh bits of randomness in case Bob rejects, since otherwise Eve will soon have accumulated too much information on the key.  

Note that~\cite{DPS05,DPS14} offers a rigorously proven bound (of roughly $m$) on the number of fresh bits necessary for key refreshing. However, their notion of key refreshing is stronger than what we require: they require that the refreshed key is close to random and independent of Eve, whereas we merely require that the refreshed key is ``secure enough'' as to ensure security of the primitive, i.e., authenticity in $\QMAC$ or $\QMAC^*$, and privacy (and authenticity) in $\QENC^*$. Indeed, in our construction we do not require that the basis key is close to random, only that it is hard to guess. However, the above informal argument shows that the bound still applies. 

Similarly, one can argue that in any message authentication scheme with error probability $2^{-\lambda}$, by keeping everything Eve can obtain $\lambda$ bits of information on the key. Thus, in case of almost no noise, our scheme $\QMAC^*$ is optimal up to the factor $3$. 

In our constructions, the number of fresh random bits needed for the key refreshing increases with larger noise. In particular in $\QMAC^*$, $\varsigma$ will soon be the dominating term in case we increase the noise level. 
We point out that it is not clear whether such a dependency is necessary, as we briefly mention in Sect.~\ref{sec:conclusion}.

\subsection{Supporting Quantum Messages}

The approach in Sect.~\ref{sec:Enc} of extracting a (one-time pad) key also gives us the means to authenticate and/or encrypt {\em quantum} messages: we simply use the extracted key as quantum-one-time-pad key~\cite{AMTW00}, or as key for a quantum message authentication code~\cite{BCGST02}. 
However, 
when considering arbitrary quantum messages, honest users anyway need a quantum computer, so one might just as well use the scheme by Damg{\aa}rd \etal to communicate a secret key and use this key for a quantum-one-time-pad or for quantum message authentication, or resort to~\cite{GYZ16,Por16}, which additionally offer security against superposition attacks.

\subsection{Variations}\label{sec:vars}

We briefly mention a few simple variations of our schemes. The first variation is as follows. In $\QMAC$, instead of choosing $x$ uniformly at random and computing the tag $t$ as $t = \MAC(k,\msg\|x)$, we can consider a fixed tag $t_\circ \in \TAG$, and choose $x$ uniformly at random subject to $\MAC(k,\msg\|x) = t_\circ$. Since $t_\circ$ is fixed, it does not have to be sent along. In case the classical $\MAC$, as a keyed hash function, is of the form as in Proposition~\ref{prop:hashconstruction}, meaning that the tag is one-time-pad encrypted (which in particular holds for the canonical examples suggested in Sect.~\ref{sec:scheme}), then Theorem~\ref{thm:security} and Proposition~\ref{prop:seqex} still hold. 
Indeed, if $\MAC$ is of this form then the concrete choice of $t_\circ$ is irrelevant for security: if Theorem~\ref{thm:security} would fail for one particular choice of $t_\circ$ then it would fail for any choice, and thus also for a randomly chosen tag, which would then contradict Theorem~\ref{thm:security} for the original $\QMAC$. 
Similarly, in $\QMAC^*$ and $\QENC^*$ we can fix the tag $t$ and the secure sketch $s$ (and ciphertext $c$), and choose $x$ subject to the corresponding restrictions.   

A second variation is to choose the basis key $\theta$ not as a code word, but uniformly random from $\set{0,1}^n$. As a consequence, the bounds on the games analyzed in Sect.~\ref{sec:GuessingGames} change\,---\,indeed, the game analyzed in Proposition~\ref{prop:monogamy} then becomes the monogamy-of-entanglement game considered and analyzed in~\cite{TFKW13}\,---\,and therefore we get different bounds in Theorem~\ref{thm:security}, but conceptually everything should still work out. 
Our goal was to minimize the number of fresh random bits needed for the key refreshing, which corresponds to the number of bits necessary to describe $\theta$; this allows us to compare our work with~\cite{DPS05,DPS14} and show that our encryption scheme performs (almost) as good as theirs in this respect. And with this goal in mind, it makes sense to choose $\theta$ as a codeword: it gives the same guessing probability for $x$ but asks for less entropy in $\theta$. Choosing $\theta$ uniformly random from $\set{0,1}^n$ seems to be the preferred choice for minimizing the quantum communication instead, which would be a very valid objective too.  

As an interesting side remark, we observe that with the above variations, our constructions can be understood as following the design principle of the scheme originally proposed by Bennett \etal of encrypting and adding redundancy to the message, and encoding the result into BB84 qubits. 

Finally, a last variation we mention is to use the {\em six-state} encoding instead of the BB84 encoding. Since the three bases of the six-state encoding have the same so-called maximal overlap, the bounds in Sect.~\ref{sec:GuessingGames} carry over unchanged, but we get more freedom in choosing the code $\code$ in $\set{0,1,2}^n$ so that fewer qubits need to be communicated for the same amount of entropy in $x$. Also, when choosing the bases uniformly at random in $\set{0,1,2}^n$, as in the variation above, we get a slightly larger entropy for $x$ when using the six-state encoding.

\section{Conclusion, and Open Problems}\label{sec:conclusion}

We reconsider one of the very first problems that was posed in the context of quantum cryptography, even before QKD, and we give the first solution that offers a rigorous security proof {\em and} does not require any sophisticated quantum computing capabilities from the honest users. 
However, our solution is not the end of the story yet. An intriguing open problem is whether it is possible to do the error correction in a more straightforward way, by just sending the syndrome of $x$ with respect to a {\em fixed} suitable code, rather than relying on the techniques from \cite{DS05}.  
In return, the scheme would be simpler, it could take care of more noise\,---\,Dodis and Smith are not explicit about the amount of noise their codes can correct but it appears to be rather low\,---\,and, potentially, the number of fresh random bits needed for key refreshing might not grow with the amount of noise. 
Annoyingly, it {\em looks} like our scheme should still be secure when doing the error correction in the straightforward way, but our proof technique does not work anymore, and there seems to be no direct fix. 

From a practical perspective, it would be interesting to see to what extent 
it is possible to optimize the quantum communication rather than the key refreshing, e.g., by using BB84 qubits with fully random and independent bases, and whether is it possible to beat QKD in terms of quantum communication. 

\section*{Acknowledgments}

We would like to thank Ivan Damg{\aa}rd and Christian Schaffner for early discussions on the topic, and Christopher Portmann for various discussions and for comments on previous versions of the paper.

\begin{appendix}

\section*{APPENDIX}

\section{Yet Another (Version of the) Guessing Game}\label{app:monogamy+}

We consider a variant of the guessing game from Section~\ref{sec:GuessingGames} where Bob and Charlie need to guess Alice's measurement outcome. In the variation considered here, we give some slack to Bob in that it is good enough if his guess is close enough (in Hamming distance) to Alice's measurement outcome, and Charlie is given some (deterministic) classical side information on Alice's measurement outcome before he has to announce his guess.%
\footnote{Taking care of such side information, given to Charlie, on Alice's measurement outcome is not needed for our application, but we get it almost for free. }
We show that, if the minimal distance $d$ of the code $\code$ is large enough, this does not help Bob and Charlie significantly. This is in line with the intuition that, for large enough $d$, the optimal strategy for Bob and Charlie is to pre-guess Alice's choice of bases.

\begin{proposition}\label{prop:monogamy+}
Let $\H_\rA$ be a $n$-qubit system, and let $\H_\rB$ and $\H_\rC$ be arbitrary quantum systems. Also, let $0 \leq \noise \leq \frac12$ be a parameter and $f:\set{0,1}^n \to {\cal Y}$ a function. Consider a state $\rho_{\Theta \rA\rB\rC} = \mu_{\code} \otimes \rho_{\rA\rB\rC} \in \dens(\code \otimes \H_\rA \otimes \H_\rB \otimes \H_\rC)$,
and let 
$$
\rho_{\Theta X X' X''} =  {\cal N}_{\Theta f(X) \rC \to X''} \circ {\cal N}_{\Theta \rB \to X'}\circ {\cal M}^\BB_{\Theta \rA \to X}\bigl(\rho_{\Theta\rA\rB\rC}\bigr)
$$
where ${\cal N}_{\Theta \rB \to X'}$ is an arbitrary measurement of system $\rB$ controlled by $\Theta$, and ${\cal N}_{\Theta f(X) \rC \to X''}$ is an arbitrary measurement of system $\rC$ controlled by $\Theta$ and $f(X)$. \\[0.3ex]
Then, it holds that
$$
P[X' \!\epsclose\! X \wedge X'' \!=\! X] \leq \frac{1}{|\code|} + \frac{2^{h(\noise)n} \cdot |{\cal Y}|}{2^{d/2}} \enspace,
$$
where $h$ is the binary entropy function. 
\end{proposition}

\begin{proof}
Here,  we can write 
$$
P[X' \!\epsclose\! X \wedge X'' \!=\! X] = \frac{1}{|\code|} \biggl\| \sum_\theta \tilde\Pi^\theta \biggr\| \leq \frac{1}{|\code|} \sum_\delta \max_\theta \bigl\| \tilde\Pi^\theta \tilde\Pi^{\theta \oplus \delta} \bigr\|
$$
for projectors 
$$
\tilde\Pi^\theta = \sum_x H^\theta \proj{x}H^\theta \otimes \bigg( \sum_{e \in B_\noise^n} P^\theta_{x \oplus e} \bigg)  \otimes Q^{\theta,f(x)}_x \enspace,
$$ 
where $B_\noise^n \subset \set{0,1}^n$ denotes the set of stings with Hamming weight at most $\noise n$. 
For any $\theta \neq \theta' \in \code$, we can upper bound $\tilde\Pi^\theta$ and $\tilde\Pi^{\theta'}$ by 
$$
\tilde\Pi^\theta \leq \tilde\Gamma^\theta := \sum_x H^\theta \proj{x}H^\theta \otimes \bigg( \sum_{e \in B_\noise^n} P^\theta_{x \oplus e} \bigg)   \otimes \I = \sum_{e \in B_\noise^n} \Gamma_e^\theta
$$
and
$$
\tilde\Pi^{\theta'} \leq \tilde\Delta^{\theta'} := \sum_x H^{\theta'} \proj{x}H^{\theta'} \otimes \I \otimes \bigg( \sum_{y \in \cal Y} Q^{\theta',y}_x \bigg) = \sum_{y \in \cal Y} \Delta_y^{\theta'}
$$
respectively, where $\Gamma_e^\theta$ and $\Delta_y^{\theta'}$ are like $\Gamma^\theta$ and $\Delta^{\theta'}$, as defined in the proof of Proposition~\ref{prop:monogamy}, for certain concrete choices of the POVM's $\set{P^\theta_x}_x$ and $\set{Q^{\theta'}_x}_x$ that depend on $e$ and $y$, respectively. As such, we get that 
$$
\bigl\| \tilde\Pi^\theta \tilde\Pi^{\theta'} \bigr\| \leq \bigl\| \tilde\Gamma^\theta \tilde\Delta^{\theta'} \bigr\| \leq \sum_{e,y} \bigl\| \Gamma_e^\theta \Delta_y^{\theta'} \bigr\| \leq \frac{|B_\noise^n| \cdot |{\cal Y}|}{2^{d/2}} \leq \frac{2^{h(\noise)n} \cdot |{\cal Y}|}{2^{d/2}} \enspace.
$$
Since we still have that $\bigl\| \tilde\Pi^\theta \tilde\Pi^\theta \bigr\| = \bigl\| \tilde\Pi^\theta \bigr\| = 1$, the claim follows. 
\qed
\end{proof}

By means of the techniques from Section~\ref{sec:GuessingGames}, we can extend the result to the case where Bob and Charlie have a-priori quantum side information on Alice's choice of bases.

\begin{corollary}\label{cor:monogamy+}
Let $\H_\rA$ be a $n$-qubit system, and let $\H_\rB, \H_\rC$ and $\H_\rE$ be arbitrary quantum systems. Also, let $0 \leq \noise \leq \frac12$ be a parameter and $f:\set{0,1}^n \to {\cal Y}$ a function. Consider a state $\rho_{\Theta \rE} \in \dens(\code \otimes \H_\rE)$, and let 
$$
\rho_{\Theta\rA\rB\rC} = {\cal Q}_{\rE\to\rA\rB\rC} \bigl(\rho_{\Theta \rE}\bigr) \in \dens(\code \otimes \H_\rA \otimes \H_\rB \otimes \H_\rC)
$$
where ${\cal Q}_{\rE\to\rA\rB\rC}$ is a CPTP map acting on $\rE$, 
and let 
$$
\rho_{\Theta X X' X''} =  {\cal N}_{\Theta f(X) \rC \to X''} \circ {\cal N}_{\Theta \rB \to X'} \circ {\cal M}^\BB_{\Theta \rA \to X}\bigl(\rho_{\Theta\rA\rB\rC}\bigr)
$$
as in Proposition~\ref{prop:monogamy+} above. Then, it holds that
$$
P[X' \!\epsclose\! X  \wedge X'' \!=\! X] \leq \guess(\Theta|\rE) + \frac{\guess(\Theta|\rE) \cdot |\code| \cdot 2^{h(\noise)n} \cdot |{\cal Y}|}{2^{d/2}} \enspace.
$$
\end{corollary}

\begin{remark}\label{rem:honestBob}
In line with the remarks in Section~\ref{sec:GuessingGames}, if Bob ``measures correctly'' but is still given some slack, and, say, Charlie is given no side information on Alice's outcome, the bound relaxes to 
$$
P[X' \!\epsclose\! X  \wedge X'' \!=\! X] \leq \guess(\Theta|\rE) + \frac{\guess(\Theta|\rE) \cdot |\code| \cdot 2^{h(\noise)n}}{2^{d}} \enspace.
$$
\end{remark}

\section{On the Existence of Suitable Secure Sketches}\label{sec:SSConstruction}

In~\cite[Lemma~5]{DS05}, Dodis and Smith show that for any constant $0 < \lambda < 1$, there exists an explicitly constructible family of binary linear codes $\set{\code_i}_{i \in \cal I}$ in $\set{0,1}^n$ with dimension $k$ that efficiently correct a constant fraction of errors and have square {\em bias} $\delta^2 \leq 2^{-\lambda n}$. Their Lemma 4 then shows that the keyed hash function $\Ext: {\cal I} \times \set{0,1}^n \to {\cal SYN} = \set{0,1}^{n-k}$, $(i,x) \mapsto syn_i(x)$ is a strong extractor, where $syn_i(x)$ is the syndrome with respect to the code $\code_i$. 
More precisely, the generalization of their result to quantum side information by Fehr and Schaffner~\cite{FS08} shows that if $\rho_{IX\rE} = \mu_{\cal I} \otimes \rho_{X \rE} \in \dens({\cal I} \otimes {\cal X} \otimes \H_\rE)$ then 
$$
\delta(\rho_{\Ext(I,X) I \rE}, \mu_{\cal SYN} \otimes \rho_I \otimes \rho_\rE) \leq \frac12 \sqrt{\guess(X|\rE) \, \delta^2 \, 2^n } \enspace.
$$
It follows from Proposition~\ref{prop:hashconstruction} that the secure sketch 
$$
\SS: {\cal L} \times \set{0,1}^n \to {\cal SYN}, \, (i\|b,x) \mapsto syn_i(x)+b
$$
where ${\cal L} := {\cal I} \times {\cal SYN}$, offers uniformity and $\nu$-key-privacy with parameter $\nu = \delta 2^{n/2}/\sqrt{|{\cal SYN}|} = \delta 2^k$. 

\smallskip

Dodis and Smith are not explicit about the {\em size} $n-k$ of the syndrome in their construction, but looking at the details, we see that $n-k\leq \log(\delta^2 \, 2^n)$. As such, by artificially extending the range ${\cal SYN} = \set{0,1}^{n-k}$ of $\SS$ to a set ${\cal S} = \set{0,1}^{\varsigma}$ of bit strings of size $\varsigma := \log(\delta^2 \, 2^n)$, and re-defining $\SS$ to map $(i\|b,x)$ to $syn_i(x)+b$ {\em padded with sufficiently  many $0$'s}, we get that the secure sketch $\SS: {\cal L} \times \set{0,1}^n \to {\cal S}$ is message-independent and offers {\em ideal} key-privacy.%
\footnote{Alternatively, we could simply stick to $\SS: {\cal L} \times \set{0,1}^n \to {\cal SYN}$ but carry along the non-ideal parameter $\nu$; however, we feel that this additional parameter would make things more cumbersome\,---\,but of course would lead to the same end result. }

\end{appendix}

\end{document}